\documentclass[12pt]{article}
\usepackage{multirow}
\usepackage{amsmath}
\usepackage{amssymb}
\usepackage{amsthm}

\usepackage{comment}
\usepackage{algorithm}  
\usepackage{xcolor}

\usepackage{graphicx} 

\usepackage{color}
\usepackage{tikz}
\usepackage{booktabs}
\usepackage{multirow}
\usepackage{enumitem}
\usepackage{url}

\usepackage{natbib}
\usepackage{multibib}

\usepackage{hyperref}
\numberwithin{equation}{section}

\newcites{SM}{References}

\usepackage{pgfplots}
\pgfplotsset{compat=1.18}
\usetikzlibrary{arrows.meta}

		\newtheoremstyle{Format1}
		{\topsep}   
		{\topsep}   
		{\normalfont}  
		{0pt}       
		{\bfseries} 
		{\newline}  
		{5pt plus 1pt minus 1pt} 
		{}          
		      
		\newtheorem{Def}{Definition}[section]

	    \newtheorem{lemma}[Def]{Lemma}

        \newtheorem{theorem}[Def]{Theorem}
	    	
	    \newtheorem{rem}[Def]{Remark}

        \newtheorem{proposition}[Def]{Proposition}

        \newcommand{\tod}{\xrightarrow{d}}
	\newcommand{\topr}{\xrightarrow{\mathbb{P}}}
        \newcommand{\ton}{\xrightarrow{n \to \infty}}
	\newcommand{\sgn}{\operatorname{sgn}}

	\newcommand{\R}{\mathbb{R}}

\title{Testing for subgroup treatment effect consistency in the Cox model}
\author{Lukas Koletzko, Holger Dette, Björn Bornkamp, Frank Bretz, Zoe Kristin Lange}

\begin{document}

  \author{
  Lukas Koletzko$^{1}$, ~~~
  Holger Dette$^{1}$, ~~~
  Björn Bornkamp$^{2}$, \\
  Frank Bretz$^{2,3}$, ~~~
  Zoe Kristin Lange$^{1}$
  }

\maketitle

  \begin{flushleft}
    $^{1}$\texttt{Faculty of Mathematics, Ruhr-Universität Bochum, Germany}\\
     $^{2}$\texttt{Statistical Methodology, Novartis Pharma AG, Basel, Switzerland}\\
     $^{3}$\texttt{Medical University of Vienna, Center for Medical Data Science, Institute of Medical Statistics, Vienna, Austria}
  \end{flushleft}
  
\begin{abstract}
An overall treatment effect in a clinical trial may inadequately represent particular patient subgroups, creating uncertainty about whether a population-level efficacy conclusion can legitimately be transferred to them. Conventional interaction tests only investigate whether subgroup-specific treatment effects are exactly equal or not, but cannot detect whether the differences are small enough to be clinically negligible. We develop a formal framework for assessing subgroup treatment effect consistency for time-to-event outcomes within the Cox proportional hazards model. Consistency is formulated as an equivalence problem based on the (weighted) treatment-by-subgroup interaction coefficient. We consider detecting consistency between two complementary subgroups and consistency of subgroup-specific treatment effects with the overall treatment effect. For each setting, we develop a conventional two one-sided tests procedure (TOST) and a new test motivated by optimal equivalence testing for normally distributed parameters. We prove the asymptotic validity of all procedures and show empirically that the new tests are more powerful than their TOST counterparts. Finally, we apply the new methodology to a case study motivated by the CANTOS cardiovascular outcomes trial.
\end{abstract}

\noindent%

\section{Introduction}
\label{sec1}

In many clinical trial settings, the treatment effect of interest is first assessed in the overall population, followed by an assessment of whether the overall finding is also representative of clinically relevant subgroups. More specifically, one may ask whether the treatment effects in different subgroups are sufficiently close to one another, or sufficiently close to the effect in the overall population. These questions are especially important in late stage drug development when treatment effects are to be interpreted across demographic or clinically defined subpopulations, and when regulatory or scientific interest extends beyond the overall analysis result
\citep{ICH_E17_2017, EMA_Subgroups_2019}. If treatment effect consistency between certain subgroups and the overall population can be justified, then the overall trial results can reasonably be applied to those subgroups. \smallskip
\\
There is a substantial literature 
on tests for assessing the null hypothesis of
treatment effect homogeneity; see, among many others, \cite{Wang2015}, \cite{Foster2016} and \cite{Shieh2017},
together with the references therein. In these works, the problem is typically formulated in terms of testing  
the null hypothesis of no interaction between treatment and the covariates defining the subgroups.  
Other work on subgroup consistency has focused primarily on design considerations rather than direct statistical inference; see, among others, \cite{liao2018sample} and \cite{wu2020regional}. These authors investigate how the overall and subgroup-specific sample sizes should be chosen to attain a desired probability, or assurance, of concluding consistency. Such approaches typically assume that consistency holds and then quantify how likely it is to be detected under that assumption.  
\smallskip
\\
We adopt a different approach in this paper. In contrast to the aforementioned references, we consider the problem of formulating and testing for treatment effect consistency in a statistically rigorous way. More specifically, we frame the problem as an equivalence-type assessment of whether any differences between subgroup effects are small enough to be clinically negligible rather than as a classical test of whether subgroup effects are exactly the same. Such a formal assessment of subgroup consistency requires several elements. First, one must specify how consistency between subgroup-specific treatment effects is to be quantified. Second, one must define a consistency margin based on clinical relevance. Third, one needs an inferential procedure with controlled Type I and Type II error rates.  
\smallskip
\\

A motivating example is the Canakinumab Anti-inflammatory Thrombosis Outcome Study (CANTOS), a randomized, double-blind, placebo-controlled phase III cardiovascular outcomes trial evaluating canakinumab in patients at high cardiovascular risk. The trial demonstrated an overall treatment benefit for the time to nonfatal myocardial infarction, nonfatal stroke, or cardiovascular death. However, women represented only about one quarter of the study population, raising the practically important question of whether the overall treatment effect could also be regarded as representative of both women and men. This illustrates the need for formal methods that assess whether subgroup-specific treatment effects differ by no more than a prespecified clinically acceptable margin, rather than relying solely on conventional interaction tests.
\medskip

The present paper addresses formal inference on subgroup treatment effect consistency for time-to-event outcomes within the Cox proportional hazards framework, focusing on the setting of two complementary subgroups that partition the study population.
Previous work has considered equivalence testing for survival functions under the proportional hazards assumption model; see, among others, \cite{furberg2021testing} and \cite{shen2023equivalence}. These contributions, however, address a different inferential target and
do not directly consider the consistency of treatment effects across clinically relevant subgroups. 
In contrast, we characterize differences in treatment effects through a treatment-by-subgroup interaction term,  say $\beta_{TS}$, which provides a natural basis for formulating and testing consistency hypotheses. We consider two practically relevant settings: The comparison of treatment effects between subgroups and the comparison of subgroup-specific treatment effects with the treatment effect in the overall population. In both settings, we formulate subgroup consistency through the consistency hypotheses $ H_0: |\beta_{TS}| \geq \theta_c$ versus $H_1: |\beta_{TS}| < \theta_c$ or weighted versions thereof,  where $\theta_c$ denotes a prespecified consistency margin defining the largest clinically acceptable deviation. \smallskip
\\
We develop two distinct test procedures for both of these settings. The first is a TOST-type procedure based on the interval inclusion principle, extending to time-to-event settings earlier results by \cite{grill2020assessing} for binary outcomes using logistic regression models and by \cite{Ring2018} for generalized linear models. The second is a novel procedure motivated by recent developments in equivalence testing \citep{dette2018equivalence,mollenhoff2020equivalence,binder2022similarity}. 
Exploiting the asymptotic normality of the partial likelihood estimator, we construct a consistency test that mimics the uniformly most powerful test for interval hypotheses concerning the mean of independent normally distributed data with known variance. 
We establish the asymptotic validity of the proposed procedure and show that, under the stated conditions, it is more powerful than the corresponding TOST-type procedure.
\smallskip
\\
We investigate the finite-sample operating characteristics of both procedures in a simulation study, with particular emphasis on Type I error control and power. The results indicate that the proposed procedure can provide substantial power gains over the TOST-type procedure, including in small-sample settings in which the nominal Type I error rate is adequately maintained. Finally, we illustrate the methodology using the motivating CANTOS cardiovascular outcomes study described above. Proofs and further technical details are provided in the Supplementary Material.

\section{Consistency of treatment effects for two subgroups}  \label{sec2} 
 We consider clinical trials conducted over a study period $[0, \tau]$, with two treatment groups (treatment and control), two complementary subgroups $\mathcal{S}_0$ and $\mathcal{S}_1$ and time to an event of interest as the primary endpoint. Let $n_\ell$ denote the number of patients in subgroup $\mathcal{S}_\ell$, $\ell = 0,1$, and let $n = n_0 + n_1$ denote the total sample size. The observed patient data are modeled by the random variables
$$\big \{ \big ( T_{{\rm obs},i},\delta_i , X_i \big )  : i=1, \ldots , n \big\}, $$
where, for patient $i=1,\ldots, n$, $T_{{\rm obs},i}= \min (T_i,C_i)$ denotes the observed follow-up time, $T_i$ the actual survival time (a positive, continuous random variable), $C_i$ the random censoring time, and $\delta_i= I \{  T_i \leq C_i \}$ (such that 1 = event and 0 = censored). We assume independent survival times $T_1, \ldots , T_n$ and i.i.d. 
censoring times $C_1, ... , C_n$, where the censoring times are independent from the survival times. We characterize patients via binary covariates modeled by i.i.d. random variables $X_1, \ldots , X_n$, where
 $X_i = (X_{iT}, X_{iS})^\top$, $X_{iT}$ decodes the treatment group membership (0 = control, 1 = treatment) and $X_{iS}$ decodes the subgroup membership (e.g., 0 = male and 1 = female) of the $i$th patient. We assume a Cox proportional hazards model with an interaction coefficient, so that the hazard function for the $i$th individual with covariates $X_i=(X_{iT},X_{iS})^\top$ can be expressed as
\begin{align}
    \label{hdn1}
h (t \mid X_i = x_i) = h_0(t) \cdot \exp \big (\beta_T x_{iT} + \beta_S x_{iS} +   \beta_{TS} x_{iT}x_{iS} \big )~,
\end{align}
where $x_i=(x_{iT},x_{iS})^\top$,  $h_0(t)$ is an unknown baseline hazard function and $\beta=(\beta_T,\beta_S,\beta_{TS})^\top $ is an unknown parameter vector  with $\beta_T$ representing the treatment effect, $\beta_S$ the subgroup effect and $\beta_{TS}$ the treatment-by-subgroup interaction effect \cite[see, for example,][]{kleinbaum2012survival}. \smallskip

 The unknown parameter $\beta$ in the Cox model \eqref{hdn1} can be estimated by maximizing the partial log-likelihood function
\begin{equation}
\label{d1}
\ell(\beta)
=
\sum_{i=1}^n
\delta_i\Big [
\beta^\top \tilde X_i
-
\log\!\Big (
\sum_{j \in R(T_{\mathrm{obs},i})}
\exp\!\big (\beta^\top \tilde X_j\big )
\Big )
\Big ],
\end{equation}
   where $ \tilde X_i = (X_{iT},X_{iS},X_{iT} \cdot X_{iS})^\top $ and
 $$R(t) = \big  \{ j  \in  \{1, \ldots , n\} ~|~ T_{{\rm obs},j} \geq t \big \} $$
denotes the risk set at time \( t \). For large sample sizes, the distribution of the partial log-likelihood estimator \( \hat{\beta} \) 
maximizing \eqref{d1} 
can be approximated  
by a $3$-dimensional normal distribution with mean $\beta$ and 
covariance matrix $\mathcal{I}^{-1}(\hat{\beta} )$ \cite[see, for example,][]{andersen1982cox}.
That is,
\begin{equation}
    \label{deta}
\hat{\beta} ~\underset{n \to \infty}{\sim}~
 \mathcal{N}_3 \big (\beta, \mathcal{I}^{-1}(\hat{\beta} ) \big )~, \quad 
\end{equation}
where 
\begin{align}
\label{fisher_est}
    \mathcal{I}(\hat{\beta} ) &= -\sum_{i=1}^n \delta_i \bigg ( 
\frac{\sum_{l \in R(T_{{\rm obs},i})} \tilde  X_l  \tilde  X_l^\top \exp(  \tilde X_l^\top \hat{\beta} )}{\sum_{l \in R(T_{{\rm obs},i})} \exp( \tilde  X_l^\top \hat{\beta} )}  \\
& ~~~~~~~~~~~~~~~ - 
\frac{\big ( \sum_{l \in R(T_{{\rm obs},i})}  \tilde X_l \exp( \tilde X_l^\top \hat{\beta} ) \big )\big ( \sum_{l \in R(T_{{\rm obs},i})}  \tilde X_l^\top \exp( \tilde X_l^\top \hat{\beta} ) \big )}{\big( \sum_{l \in R(T_{{\rm obs},i})} \exp( \tilde X_l^\top \hat{\beta} ) \big )^2}
\bigg )
\nonumber
\end{align}
denotes the plug-in estimator of the observed information matrix $\mathcal{I}(\beta)$ \citep[][, Chapter 3]{therneau2000modeling}.

\bigskip

In this section, we directly assess treatment effect consistency between the two complementary subgroups $\mathcal{S}_0$ and $\mathcal{S}_1$, without reference to the treatment effect in the overall population. The results provide a foundation for Section \ref{sec3}, in which the subgroup-specific treatment effects are compared with the overall treatment effect, while also being of independent methodological interest.

\subsection{Consistency hypotheses}

We define the treatment effect at time $t$ for a subgroup $\mathcal S_0$, specified by $X_S = 0$, and for its complement $\mathcal S_1$, specified by $X_S = 1$, in terms of the conditional $\log$-hazard ratio
\begin{align}
\label{subgroupeffects}
    &   \text{HR}_{s}(t) := \log\Big(\frac{h(t \mid X_T=1, X_S=s)}{h(t \mid X_T=0, X_S=s)} \Big), \quad s=0,1.
\end{align}
Under 
the proportional hazards model \eqref{hdn1} the effects  become time-independent and independent of the baseline hazard function $h_0$ and reduce to linear functions of only the treatment effect $\beta_T$ and the treatment-by-subgroup interaction $\beta_{TS}$:
\begin{align}
\label{subgroupeffects_2}
    &
    \text{HR}_{s}
    =
    \beta_T + s \cdot\beta_{TS}, \quad s=0,1.
\end{align}

As a consequence, we can investigate  the consistency assessment  of the  treatment effects by testing the hypotheses 
\begin{align}
\label{hypothesis1}
    H_0: |\beta_{TS}| \geq \theta_c \quad \text{against} \quad H_1: |\beta_{TS}| < \theta_c.
\end{align}
Note that $\beta_{TS}$ is the additive change in treatment effect on the log-scale between the two subgroup hazard ratios. Therefore, if the consistency margin $\theta_c > 0$ is chosen appropriately, the alternative hypothesis $H_1$ can be interpreted to mean that the two subgroup treatment effects are consistent, i.e., the treatment effect difference between the two subgroups is clinically irrelevant.

 \smallskip

\begin{rem} ~~~~
 
{\rm 
\begin{itemize}
\item[(i)]  
The definition of the hypotheses in \eqref{hypothesis1} does  not take into account the relation of the additive change in the treatment effect  $\beta_{TS} $ to the treatment effect $\beta_T$. That is, a change in the treatment effect  of size $\beta_{TS} =0.5$  may be considered as small, if the treatment effect is $\beta_T=5$ and as large if 
$\beta_T=0.2$. To address this issue one can also investigate 
 the ratio
$\frac{\beta_{TS}}{\beta_{T}}$ of the subgroup-specific
treatment effect difference and the overall treatment effect for the quantification of the
interaction  
\cite[see, for example,][]{Brookes2001,Ring2018,grill2020assessing}. The corresponding consistency hypotheses are then given by 
\begin{align*}
    H_0: \Big | \frac{\beta_{TS}}{\beta_{T}}  - 1 \Big | \geq \theta_c \quad \text{against} \quad H_1: \Big | \frac{\beta_{TS}}{\beta_{T}}  - 1 \Big |  < \theta_c.
\end{align*}
The methodology developed in this section for testing the hypotheses \eqref{hypothesis1} can be easily extended for testing these hypotheses. 
\item[(ii)]
We emphasize that hypotheses \eqref{hypothesis1} are fundamentally different from the more well-known classical problem addressed by interaction tests
\begin{align} \label{hd1}
    H_0: \beta_{TS} = 0 \quad \text{against} \quad H_1: \beta_{TS} \neq 0,
\end{align}
which is often used as an alternative 
for investigating the hypothesis of subgroup consistency. 
This classical formulation, however,  has several disadvantages. First, if $H_0$ in \eqref{hd1} is not rejected the Type II error of  deciding for 
exact consistency (that is 
$\beta_{TS} = 0$) is unknown. Second, exact equality of the two subgroup treatment effects is very unlikely in medical practice and consequently one would test a null hypothesis 
in \eqref{hd1}, which is believed not to be true anyways.  
Third, if small differences exist, the rejection of the  null hypothesis is guaranteed as soon as the sample size is sufficiently large. Therefore,  the null hypothesis is correctly rejected, although the difference might not  be of clinical relevance. 
\end{itemize}
}
\end{rem}

\subsection{Testing methodology} \label{sec2.2}

Note that the consistency hypotheses \eqref{hypothesis1} coincide with the hypotheses (10.25) in \cite{Ring2018} and the hypotheses (11) in \cite{grill2020assessing}. Applying the classical two one-sided tests (TOST) approach followed by these authors, we can investigate  these hypotheses  by testing the two one-sided pairs of hypotheses 
\begin{align*}
    H_0^{(1)}: \beta_{TS} \geq \theta_c  \quad &\text{against} \quad 
     H_1^{(1)}: \beta_{TS} < \theta_c \\
    \text{and} \quad H_0^{(2)}: \beta_{TS} \leq -\theta_c  \quad &\text{against} \quad
     H_1^{(2)}: \beta_{TS} > -\theta_c
\end{align*}
simultaneously using a combination of the intersection union principle (see \cite{berger1982}) and 
two one-sided tests, which are obtained from the duality between confidence intervals and tests (see \cite{Aitchison}). 
More specifically, 
the  test based on the  TOST rejects the null hypothesis in  \eqref{hypothesis1} whenever 
\begin{align}
\label{TOST1}
\big[ \hat \beta_{TS} - z_{1- \alpha } \, \hat \sigma_{\hat \beta_{TS}}, \ 
\hat \beta_{TS} + z_{1- \alpha} \, \hat \sigma_{\hat \beta_{TS}} \big]     
    \subseteq (-\theta_c, \theta_c), 
\end{align}
where $\hat \sigma_{\hat \beta_{TS}}^2 := {\big[\mathcal I(\hat \beta)^{-1}\big]_{33}}$ denotes an estimator of the asymptotic variance of 
the partial likelihood estimator $\hat \beta_{TS}$, $\mathcal{I}(\hat \beta)$ the plug-in estimator of the observed information matrix defined in \eqref{fisher_est} and $z_{1-\alpha}$ the $(1-\alpha)$-quantile of the standard normal distribution. 
The $p$-value of the test \eqref{TOST1} can be approximated by 
    \begin{align} \label{ptost}
        p_{TOST} = \max 
        \Big \{ \Phi \Big(\frac{\hat \beta_{TS} - \theta_c}{\hat \sigma_{\hat \beta_{TS}}} \Big),  \Phi \Big(\frac{-\hat \beta_{TS} - \theta_c}{\hat \sigma_{\hat \beta_{TS}}} \Big) \Big\} ,
    \end{align}
    where $\Phi $ denotes the cumulative distribution function (CDF) of the standard normal distribution.
We prove in Theorem \ref{thm1} in the Supplementary Material that the  decision rule \eqref{TOST1}  yields an asymptotically valid test. 
Procedures based on the TOST principle, including the test defined in \eqref{TOST1}, can nevertheless be conservative due the application of the intersection union principle. Despite this potential conservatism, TOST-based procedures are widely used in practice because of their simplicity. 

\smallskip

In the following, we present 
a novel and more powerful
method for testing the hypotheses \eqref{hypothesis1} which  mimics  the uniformly most powerful test for the hypotheses $H_0: | \mu | \geq \epsilon $ versus $H_1: | \mu | < \epsilon$
for univariate normally distributed data with mean $\mu$ and  known variance \citep[see][]{romano}.
As  the estimator $\hat \beta_{TS}$ is asymptotically normal distributed and the asymptotic variance can be estimated consistently, 
this test  provides a more powerful approach for testing the hypotheses  \eqref{hypothesis1} than the test \eqref{TOST1} based on the TOST principle. We give a rigorous argument for this superiority in Remark \ref{rem1} below and also demonstrate  the finite-sample gains of the new  method  by means of a simulation study in Section \ref{sec4}.

\smallskip

\begin{algorithm}[H]
\caption{Powerful  test for hypotheses \eqref{hypothesis1}}
\label{alg1}

    \begin{itemize}
    \item[(1)] Calculate the estimator  $\hat \beta$  by maximizing the partial log-likelihood function in \eqref{d1}. 
    \item[(2)] Calculate the cond. $\alpha$-quantile $q_\alpha^{*(1)}$ 
     of the folded normal distribution $\big |\mathcal{N} \big( \theta_c , \ \hat \sigma_{\hat \beta_{TS}}^2 \big)\big |$.
    \item[(3)] Reject the null hypothesis in \eqref{hypothesis1} whenever
     \begin{align}
        \label{bootstrap1}
       |\hat \beta_{TS}|  < q_\alpha^{*(1)}. 
     \end{align}
\end{itemize}
\end{algorithm}

The quantile $q_\alpha^{*(1)}$ in \eqref{bootstrap1} can be determined as the unique solution $y^* \geq 0$ of the equation
$$
\Phi\!\Big (\frac{y^*-\theta_c }{\hat \sigma_{\hat \beta_{TS}}}\Big)
-
\Phi\!\Big (\frac{-y^*-\theta_c }{\hat \sigma_{\hat \beta_{TS}}}\Big )
= \alpha
$$
and the corresponding $p$-value can be approximated by
\begin{align} \label{pboot}
p_{alg} =
\Phi\!\Big(\frac{|\hat \beta_{TS}| -\theta_c }{\hat \sigma_{\hat \beta_{TS}}}\Big )
-
\Phi\!\Big  (\frac{-|\hat \beta_{TS}| -\theta_c }{\hat \sigma_{\hat \beta_{TS}}}\Big ).
\end{align}
The validity of the  test \eqref{bootstrap1} is proved in Theorem \ref{thm2} in the Supplementary Material.
\smallskip

  \begin{rem} ~~~
  \label{rem1}
      {\rm 
(a) Comparing the approximations of the $p$-values in \eqref{ptost} and \eqref{pboot}
     we observe that $p_{alg} <  p_{TOST}$, which indicates that for large sample sizes the  test \eqref{bootstrap1} has more power than the test based on the TOST principle.
This is also confirmed  by means of a simulation study  in Section \ref{sec4}, which shows that for realistic sample sizes  the new  test has   always larger power than the TOST-approach. This  improvement is particularly notable for small sample sizes and an at least moderate   censoring rate, where the TOST completely breaks down. 

(b) An important element for testing the hypotheses \eqref{hypothesis1} is the specification of the consistency margin $\theta_c$. This threshold  depends sensitively on the particular application and, ideally, has to be specified  in close collaboration with the clinical team. If such a specification  
is not possible, we can also fix the consistency threshold in a data-driven way. More specifically, note that the hypotheses \eqref{hypothesis1}
are nested and the decision rules in \eqref{TOST1} and \eqref{bootstrap1} are  monotone with respect to $\theta_c$. Therefore, by the sequential rejection principle \citep[see][]{Sonnemann2008}, we can simultaneously test the hypotheses \eqref{hypothesis1} and find the smallest threshold $\hat \theta_c(\alpha)$ for which the null hypothesis is rejected at significance level $\alpha$. This value can also be interpreted as a measure of evidence for similarity with a controlled Type I level $\alpha$. \smallskip
 }
  \end{rem}

\section{Consistency of subgroup-specific treatment effects with the overall population treatment effect}
\label{sec3}

In this section, we assess the consistency of subgroup-specific treatment effects with the overall population treatment effect. We consider both the comparison of a single subgroup-specific effect with the overall effect and the simultaneous comparison of both subgroup-specific effects with the overall effect.

\subsection{Consistency hypotheses}

Given the subgroup treatment effects $\text{HR}_0$ and $\text{HR}_1$ in \eqref{subgroupeffects_2}, we define the population treatment effect $\text{HR}_{pop}$ as the \textit{expected} subgroup log-hazard ratio, that is,
\begin{align}
\label{popeffect}
    \text{HR}_{pop} &  := \mathbb E_{X_S}\big[ \text{HR}_{X_S}  \big] = (1-p) \ \text{HR}_{0}  + p \ \text{HR}_{1}  = 
     \beta_T + p \beta_{TS},
\end{align}
where $X_S$ is a $\text{Bernoulli}(p)$-distributed random variable defining if a patient belongs to subgroup ${\cal S}_0$  ($X_S=0)$ or ${\cal S}_1$ ($X_S=1)$ and $p \in (0,1)$ represents the unknown proportion of subgroup $\mathcal{S}_1$ in the population. We emphasize that $p$ does not denote the trial-specific subgroup prevalence (which is $n_1/n$ and so known), but instead the population-specific subgroup prevalence (which may be unknown). 

\begin{rem} {\rm 
The population treatment effect $\text{HR}_{pop}$ in \eqref{popeffect} is a weighted average of the subgroup treatment effects and we point out that similar concepts have been used in the literature. For example, \cite{grill2020assessing} and \cite{Ring2018} define the population treatment effect $\Delta = \pi_1 \delta_1 + \pi_2 \delta_2$, where $\pi_1, \pi_2$ denote the proportions of two complementary subgroups in the population and $\delta_1$, $\delta_2$ denote the subgroup treatment effects.
Similarly, \cite{dette2025testing} define an overall dose-response relationship for the population by $\bar \mu = \sum_{\ell=1}^k p_\ell \mu_\ell$, where $p_1, ... , p_k$ denote the known proportions of $k$ disjoint subgroups partitioning the population and $\mu_1, ... , \mu_k$ represent the  dose-response curves in the subgroups.} 
\end{rem}

\smallskip

 The absolute difference between the subgroup treatment effects and the population treatment effect equals 
\begin{align}
    \label{hd10}
\big  | \text{HR}_{pop} -  \text{HR}_s \big |  =   \left\{ 
    \begin{array}{ccc}
         &    p \cdot |\beta_{TS}| & \text{ if } s=0 \\
         &    (1-p) \cdot |\beta_{TS}| & \text{ if } s=1
    \end{array}\right. 
\end{align}
such that we can investigate the consistency assessment between the subgroup ${\cal S}_0$ and the population by testing the hypotheses
\begin{align}
    \label{weighted_hyps1}
    H_0: p \cdot |\beta_{TS}| \geq \theta_c \quad \text{against} \quad H_1:  p \cdot | \beta_{TS}| < \theta_c,
\end{align}
and the consistency assessment between the subgroup ${\cal S}_1$ and the population by testing the hypotheses
\begin{align}
    \label{weighted_hyps2}
    H_0: (1-p) \cdot |\beta_{TS}| \geq \theta_c \quad \text{against} \quad H_1: (1-p)\cdot | \beta_{TS}| < \theta_c.
\end{align}

\smallskip

To investigate the consistency between both subgroups $\mathcal{S}_0$ and ${\cal S}_1$ and the population we propose testing the hypotheses
\begin{align}
    \label{weighted_hyp}
    H_0: \max \{p,1-p\} \cdot |\beta_{TS}| \geq \theta_c \quad \text{against} \quad H_1:  \max \{p,1-p\} \cdot | \beta_{TS}| < \theta_c . 
\end{align}
If the consistency margin $\theta_c > 0$ and the nominal  are  chosen appropriately, rejecting the null hypothesis in \eqref{weighted_hyps1} or \eqref{weighted_hyps2} means that the treatment effect difference between the respective subgroup identified by $s$  is consistent with the overall population treatment effect.  In the same way, rejecting the null hypothesis in \eqref{weighted_hyp}  means that  both subgroup treatment effects are consistent with the population treatment effect. 

\smallskip

In order to develop testing methodology for these hypotheses,  we  assume that the unknown subgroup prevalence $p$ can be estimated by an estimator $\hat p$ satisfying the following assumption.
\begin{itemize}
    \item[(P1)]  For large sample sizes the joint distribution of $(\hat \beta_{TS}, \hat p)^\top $ can be approximated by a normal distribution, that is 
\begin{align}
  (\hat \beta_{TS}, \hat p)^\top  \underset{n \to \infty}{\sim}~
 \mathcal{N}_2 ( ( \beta_{TS},  p)^\top  , \hat \Sigma ),
\label{hd11}
\end{align}
    where $ \hat \Sigma $ is an appropriate estimator of the asymptotic covariance.
\end{itemize}
In Section \ref{sec34} we discuss estimators $\hat p$ satisfying  Assumption (P1). In this case we can develop  TOST and a more powerful methodology, as considered in the previous section, for assessing the consistency of subgroups with the population.

\subsection{Testing methodology for the hypotheses \texorpdfstring{\eqref{weighted_hyps1}}{(weightedhypotheses1)} and \texorpdfstring{\eqref{weighted_hyps2}}{(weightedhypotheses2)} \label{sec32}}

Note that the hypotheses \eqref{weighted_hyps1} and \eqref{weighted_hyps2} can be written in the form 
\begin{align}
    \label{weighted_hyp1}
    H_0: w(p)  \cdot |\beta_{TS}| \geq \theta_c \quad \text{against} \quad H_1:  w(p)  \cdot | \beta_{TS}| < \theta_c 
\end{align}
 using the weight functions $w(p) = p$ and $w(p) = 1-p$, respectively. Therefore, we can treat both cases simultaneously. Moreover, the methodology developed in this section is applicable to more general weight functions $w: \R \to (0,1) $ as well. The corresponding test based on the TOST principle rejects the null hypothesis whenever 
\begin{align}
\label{TOST2}
\hat I_n^{w(p)} := \big[ w(\hat p)\,\hat \beta_{TS} - z_{1- \alpha } \, \hat \sigma_{w(p)} , \  w(\hat p)\,\hat \beta_{TS} + z_{1- \alpha} \, \hat \sigma_{w(p)}   \big ]   
    \subseteq (-\theta_c, \theta_c),
\end{align}
where \begin{align}
\label{varestim}
\hat \sigma_{w(p)}^2 := (w(\hat p), w^\prime (\hat p) \hat \beta_{TS}) \ \hat \Sigma \ (w(\hat p), w^\prime(\hat p) \hat \beta_{TS})^\top
\end{align} 
is an estimator of the variance of the estimator $w(\hat p) \hat \beta_{TS}$, the matrix $\hat \Sigma$ is defined in \eqref{hd11} and the choice $w(p)=p$ corresponds to \eqref{weighted_hyps1} and $w(p)=1-p$ to \eqref{weighted_hyps2} (see Theorem \ref{thm5} in the Supplementary Material for the precise statement). Alternatively, the $p$
-value of the test \eqref{TOST2} can be approximated by 
    \begin{align} \label{ptost2}
        p^{(w)}_{TOST} = \max 
        \Big \{ \Phi \Big(\frac{w(\hat p) \hat \beta_{TS} - \theta_c}{\hat \sigma_{w(p)}} \Big),  \Phi \Big(\frac{-w(\hat p)\hat \beta_{TS} - \theta_c}{\hat \sigma_{w(p)}} \Big) \Big\}.
    \end{align} \smallskip

Similar to the discussion in Section \ref{sec2.2} we can construct a more powerful test than the test obtained by the TOST principle for the hypotheses in \eqref{weighted_hyps1} and \eqref{weighted_hyps2}. 
The details are given in Algorithm \ref{alghol}. 

\begin{algorithm}[H]
\caption{Powerful test for hypotheses \eqref{weighted_hyps1} and \eqref{weighted_hyps2} \label{alghol}}
\label{alg2}
    \begin{itemize}
    \item[(1)] Calculate the estimator $\hat \beta$ maximizing the partial log-likelihood function in \eqref{d1} and the estimator $\hat p$.  
    \item[(2)] For \eqref{weighted_hyps1} define $w(p)=p$, for \eqref{weighted_hyps2} define $w(p)=1-p$ and calculate the conditional $\alpha$-quantile $q_\alpha^{*(2)}$ of $|\mathcal{N}(\theta_c, \hat \sigma_{w(p)}^2)|$.
   \item[(3)] Reject the corresponding null hypothesis whenever 
     \begin{align}
         \label{bootstrap2}
       w(\hat p) \, |\hat \beta_{TS}|  < q_\alpha^{*(2)}.
     \end{align}
\end{itemize}
\end{algorithm}
The asymptotic validity of the test \eqref{bootstrap2} is proved in Theorem \ref{thm6} in the Supplementary Material. Similar to Algorithm 1 the corresponding $p$-value is approximated by 
\begin{align}
    \label{pbootw}
    p^{(w)}_{alg} = \Phi\!\Big(\frac{w(\hat p)|\hat \beta_{TS}| -\theta_c }{\hat \sigma_{w(p)} }\Big )
-
\Phi\!\Big  (\frac{-w(\hat p)|\hat \beta_{TS}| -\theta_c }{\hat \sigma_{w(p)}}\Big )
\end{align}
and we note that $p^{(w)}_{alg} < p^{(w)}_{TOST}$.

\subsection{Testing methodology for the hypotheses \texorpdfstring{\eqref{weighted_hyp}}{(weighted hypotheses)}} \label{sec33}

Note that hypotheses \eqref{weighted_hyp} coincide with hypotheses (10) in  \cite{grill2020assessing}, who used the TOST approach for the construction of a test. More specifically, they proposed to reject the null hypothesis whenever
\begin{align}
\label{FOST} 
  \text{and ~~~~~~~~~}
  \begin{split} 
  \hat I_n^{p} &:=  \big[ \hat p\,\hat \beta_{TS} - z_{1- \alpha } \, \hat \sigma_{p} ~, ~ \hat p\,\hat \beta_{TS} + z_{1- \alpha} \,\hat \sigma_{p} \big ]  \subseteq (-\theta_c, \theta_c)  \\
    \hat I_n^{1-p} &:=   \big[ (1-\hat p) \,\hat \beta_{TS} - z_{1- \alpha } \, \hat \sigma_{1-p} ~, ~ (1-\hat p) \,\hat \beta_{TS} + z_{1- \alpha} \, \hat \sigma_{1-p} \big ]  \subseteq (-\theta_c, \theta_c) 
  \end{split}
\end{align}
(a precise statement of the validity of the test can be found in Theorem \ref{thm3} in the Supplementary Material).
The $p$-value of the test \eqref{FOST} can be approximated by  
    \begin{align*}
        p_{TOST}^{(max)} = \max \Big \{ 
        \Phi \Big  (\frac{\hat p\hat\beta_{TS} - \theta_c}{\hat\sigma_{p} } \Big ),
        \Phi \Big  (\frac{-\hat p\hat\beta_{TS} - \theta_c}{\hat\sigma_{p} }   \Big ),
        \Phi \Big  (\frac{(1-\hat p) \hat\beta_{TS} - \theta_c}{\hat\sigma_{1-p} }   \Big ),
        \Phi \Big  ( \frac{-(1-\hat p)\hat\beta_{TS} - \theta_c}{\hat\sigma_{1-p} }
        \Big )  \Big \} .
    \end{align*}

In Algorithm \ref{alg3} we propose a  more powerful   test which  provides an alternative to the test \eqref{FOST} for the hypotheses \eqref{weighted_hyp} and we  illustrate its superiority in Section \ref{sec4} by means of a simulation study (see Table \ref{tab2}).

\begin{algorithm}[H]
\caption{Powerful test for hypotheses \eqref{weighted_hyp}}
\label{alg3}
    \begin{itemize}
    \item[(1)] Calculate the  estimator $\hat \beta$ maximizing the partial log-likelihood function in \eqref{d1}
    and the estimator $\hat p$.
    \item[(2)] Calculate the cond.  $\alpha$-quantile $q_\alpha^{*(3)}$ of the distribution of $\max \{ \hat p^* , 1 -\hat p^* \} \, |\hat \beta_{TS}^*|$ with 
      \begin{align} \label{norm3}
          (\hat \beta_{TS}^*, \hat p^*)^\top \sim \mathcal{N}_2\big( \ \big( \theta_c / \max \{ \hat p , 1 -\hat p \} ), \hat p \big)^\top, \ \hat \Sigma\ \big).
      \end{align}
     \item[(3)] 
     Reject the null hypothesis in \eqref{weighted_hyp} whenever 
     \begin{align}
         \label{bootstrap3}
      \max \{\hat p,1-\hat p\}  \, |\hat \beta_{TS}|  < q_\alpha^{*(3)}.
     \end{align}
\end{itemize}
\end{algorithm}

The validity of the  test \eqref{bootstrap3} is proved in Theorem \ref{thm4} in the Supplementary Material. In contrast to Algorithm \ref{alg1}-\ref{alg2}, the quantile $q_\alpha^{*(3)}$ cannot be simulated directly from the limit distribution of $\max \{\hat p,1-\hat p\}  \, |\hat \beta_{TS}|$, as this distribution differs for $p \neq 1/2$ and $p=1/2$. This has the further effect that the test \eqref{bootstrap3} may be conservative for $p=1/2$ (see Theorem \ref{thm4} for the details) and that in general there exists no closed-form expression for the $p$-value of the test. In practice, the theoretical $\alpha$-quantile $q_\alpha^{*(3)}$ can be simulated by the empirical $\alpha$-quantile of $B$ realizations $\max \{ \hat p_1^* , 1 -\hat p_1^* \} \,  |\hat \beta_{TS,1}^*|, ... , \max \{ \hat p_B^* , 1 -\hat p_B^* \} \, |\hat \beta_{TS,B}^*|$ of the bootstrap statistic $\max \{ \hat p^* , 1 -\hat p^* \} |\hat \beta_{TS}^*|$ to arbitrary precision. Similarly, the $p$-value is given by
    \begin{align}
    \label{pvalue_boot2}
    p^{(B)}_{boot} = \frac{1}{B} \sum_{b=1}^B I \big( \max \{ \hat p_b^* , 1 -\hat p_b^* \} \, | \hat \beta_{TS,b}^* | \leq \max \{ \hat p , 1 -\hat p \} \, |\hat \beta_{TS}| \big).
    \end{align}

\begin{rem} ~~~
{\rm 
        Subgroup-versus-subgroup consistency, as defined by the hypotheses in \eqref{hypothesis1}, is a stronger requirement than consistency of one or both subgroup-specific treatment effects with the overall population treatment effect, as defined by the hypotheses in \eqref{weighted_hyps1}--\eqref{weighted_hyp}. Specifically, if the two subgroup-specific treatment effects are consistent at a given margin $\theta_c$ (i.e., $|\beta_{TS}| < \theta_c$), then they will also be consistent with the overall population effect at that same margin (i.e., $\max \{ p, 1-p \} \, |\beta_{TS}| < \theta_c$). The converse does not generally hold: Even if one or both subgroup-specific treatment effects are consistent with the overall population at margin $\theta_c$, the difference between the two subgroup-specific effects may exceed that margin.
       }
\end{rem}

\subsection{Estimating the subgroup prevalence  \label{sec34}}

We briefly discuss two estimators  for the subgroup prevalence for which we prove in Section \ref{sec83} that   assumption (P1) is satisfied. In the same section, we also provide details for the corresponding variance estimator \eqref{varestim}.
Other estimators such as  the estimator  based on inverse probability-of-selection weights  proposed by \cite{cole2010generalizing} could be used as well, but we omit a more comprehensive discussion for the sake of brevity. \smallskip

 Arguably, the most natural estimator for the subgroup prevalence $p$ is the proportion of patients
$
    \hat p_{mle} = \frac{1}{n} \sum_{i=1}^n X_{iS}
$
in the trial which belong to the subgroup $\mathcal{S}_1$ and we show in Proposition \ref{prop0} in Section \ref{sec83} that $\hat p_{mle}$ 
satisfies (P1). We will use this estimator in the simulation study and the data application.
\smallskip
 
If the trial sample cannot be viewed as a representative sample from the population, an estimate is still possible if an additional covariate can be  recorded for patients in the trial sample and for patients in the non-trial sample. More specifically, following the discussion in Section 5.1 in \cite{dahabreh2020extending}, we denote the  additional covariate vector recorded for patients in the trial sample  by $Z \in \mathbb R^p$  and for patients in the non-trial sample by $Z^* (\in \mathbb R^p)$. By the law of iterated expectation, we can express $p$ as
\begin{align} \label{itexp}
    p = \mathbb E_{Z} ( \mathbb P( X_S = 1 | Z) ).
\end{align}
If $\hat p(Z)$ is an estimate of the conditional probability  $p(Z) = \mathbb P( X_S = 1 | Z) $ we can estimate the subgroup prevalence by  
\begin{align}
    \label{outcome}
   \hat p_{OM} = \frac{1}{n^*} \sum_{i=1}^{n^*} \hat p (Z_i^*)
\end{align}
(here $n^*$ be the non-trial sample size). The consistency of this estimator is not guaranteed in general and studied in \cite{dahabreh2020extending} under specific assumptions.

\section{Finite sample properties \label{sec4}}

In this section we conduct a simulation study with two objectives: first, to compare the Type I error rate and power of the TOST \eqref{TOST1} and new test \eqref{bootstrap1} for the subgroup-versus-subgroup hypotheses \eqref{hypothesis1} and second, to compare the Type I error rate and power of the TOST-based test \eqref{FOST} and new test \eqref{bootstrap3} for the subgroups-versus-population hypotheses \eqref{weighted_hyp}. We omit a comparison of the (weighted) TOST \eqref{TOST2} and new test \eqref{bootstrap2} for the sake of brevity and only note that the conclusions are very similar. \smallskip

\textbf{Setup.} We consider a generic Phase III trial for time-to-event data lasting 2 years with a recruitment window of 1 year. Survival times (measured in days) are generated from the Cox model \eqref{hdn1} with Weibull baseline hazard given by $h_0(t) = \lambda \kappa t^{\kappa - 1} $ with scale $\lambda = 0.002$, shape $\kappa = 1.2$ and parameter vector $\beta = ( -0.2 , -0.2 , \beta_{TS})^\top$. Right-censoring times due to study-end or loss-to-follow-up are generated from $C = \min(C_1, C_2)$, where $C_1 \sim 730 - U[0,365]$ (study end) and $C_2 \sim \text{Exp}(\gamma)$, $\gamma > 0$ (loss-to-follow-up). We simulate the rejection probability of the tests for the sample sizes $n = 200, 800, 5000$ corresponding to three broad categories of Phase III trials: small (e.g. rare diseases), medium (e.g. oncology) and large (e.g. cardiovascular). Patients are split into two subgroups with subgroup sample sizes denoted by $n_0$ and $n_1$. For the subgroup-versus-subgroup hypotheses \eqref{hypothesis1}, we assume the subgroup sample sizes are pre-specified and for each total sample size $n$, we consider two cases: a balanced case $(n_0 = n_1)$ and an unbalanced case $(n_0 \neq n_1)$, where $n_1 = n/4, \, n_0 = 3n/4$. For the subgroups-versus-population hypotheses \eqref{weighted_hyp}, we assume the subgroup sample sizes are random variables $n_1 \sim \text{Binomial}(n,p)$ and $n_0 = n - n_1$, where the subgroup prevalence $p$ is unknown and estimated by the trial subgroup proportion $\hat p = n_1/n$. Again, we investigate the balanced case ($p=0.5$) and an unbalanced case ($p=0.25$). For each sample size scenario (small, medium and large), we explore the performance of the tests for three different levels of censoring: low (20 \%), medium (50 \%) and high (80 \%). We fix the consistency threshold $\theta_c = \log(2)$ and vary the magnitude of the interaction coefficient $\beta_{TS}$ around $\log(2)$ in order to simulate the rejection probability of the tests in the interior of $H_0$, on the boundary of $H_0$ (``null boundary") and under the alternative $H_1$. Note that the rejection rates of the tests are mainly driven by the \textit{difference} $\theta_c - \beta_{TS}$, not by the individual values of $\theta_c$ and $\beta_{TS}$, which is why we focus on a fixed threshold $\theta_c$ throughout the simulation and refrain from considering multiple thresholds. For each choice of $\beta_{TS}$ we adapt the dropout rate $\gamma$ to maintain a certain censoring proportion ($20 \%$, $50 \%$ or $80 \%$) yielding 3 sub-tables labeled (Low), (Medium) and (High). The empirical $\alpha$-quantiles for Algorithm \ref{alg1} and Algorithm \ref{alg3} are calculated based on $B=5000$ bootstrap samples and all rejection probabilities are approximated by their empirical averages based on $5000$ test runs. \smallskip

\begin{table}[!htbp]
\footnotesize
\centering
\begin{tabular}{cc|c|c||c|c||c|c|}
\cline{3-8}
 & & \multicolumn{2}{c||}{$n=200$} & \multicolumn{2}{c||}{$n=800$} & \multicolumn{2}{c|}{$n=5000$} \\
\cline{2-8}
\multicolumn{1}{c|}{} &
\multicolumn{1}{c|}{$\beta_{TS}$} &
$n_0 = n_1$ & $n_0 \neq n_1$ &
$n_0 = n_1$ & $n_0 \neq n_1$ &
$n_0 = n_1$ & $n_0 \neq n_1$ \\
\cline{2-8}

\multicolumn{1}{c|}{\multirow{12}{*}{(Low)}} 
 & \multicolumn{1}{|c|}{\multirow{2}{*}{$\log(2.5)$}}  & \rule{0pt}{3ex} 0.010 & 0.012 & 0.001 & 0.002 & 0.000 & 0.000 \\
 & \multicolumn{1}{|c|}{} & (0.009) & (0.007) & (0.000) & (0.001) & (0.000) & (0.000) \\[4pt]
 & \multicolumn{1}{|c|}{\multirow{2}{*}{$\log(2.25)$}}  & 0.023 & 0.029 & 0.008 & 0.012 & 0.000 & 0.000 \\
 & \multicolumn{1}{|c|}{}  & (0.019) & (0.017) & ( 0.007) & (0.008) & (0.000) & (0.000) \\[4pt]
 & \multicolumn{1}{|c|}{\multirow{2}{*}{
\textbf{log(}$\mathbf{2}$\textbf{)}}} & \textbf{0.050} & \textbf{0.051} & \textbf{0.047} & \textbf{0.047} & \textbf{0.052} & \textbf{0.050} \\
 & \multicolumn{1}{|c|}{}  & \textbf{(0.046)} & \textbf{(0.032)} & \textbf{(0.053)} & \textbf{(0.053)} & \textbf{(0.053)} & \textbf{(0.051)} \\[4pt]
 & \multicolumn{1}{|c|}{\multirow{2}{*}{$\log(1.75)$}} & 0.104 & 0.083 & 0.214 & 0.184 & 0.673 & 0.559 \\
 & \multicolumn{1}{|c|}{}  & (0.101) & (0.056) & (0.207) & (0.179) & (0.669) & (0.558) \\[4pt]
 & \multicolumn{1}{|c|}{\multirow{2}{*}{$\log(1.5)$}} & 0.199 & 0.144 & 0.572 & 0.458 & 0.998 & 0.987 \\
 & \multicolumn{1}{|c|}{}  & (0.184) & (0.080) & (0.562) & (0.455) & (0.998) & (0.987) \\[4pt]
 & \multicolumn{1}{|c|}{\multirow{2}{*}{$\log(1.25)$}} & 0.334 & 0.224 & 0.905 & 0.819 & 1.000 & 1.000 \\
 & \multicolumn{1}{|c|}{}  & (0.316) & (0.159) & (0.902) & (0.818) & (1.000) & (1.000) \\[4pt]
\cline{2-8}
\cline{2-8}

\multicolumn{1}{c|}{\multirow{12}{*}{(Medium)}} 
 & \multicolumn{1}{|c|}{\multirow{2}{*}{$\log(2.5)$}} & \rule{0pt}{3ex} 0.018 & 0.020 & 0.003 & 0.004 & 0.000 & 0.000 \\
 & \multicolumn{1}{|c|}{}  & (0.004) & (0.000) & (0.002) & (0.004) & (0.000) & (0.000) \\[4pt]
 & \multicolumn{1}{|c|}{\multirow{2}{*}{$\log(2.25)$}} & 0.027 & 0.036 & 0.012 & 0.015 & 0.001 & 0.001 \\
 & \multicolumn{1}{|c|}{}  & (0.005) & (0.000) & (0.011) & (0.014) & (0.000) & (0.001) \\[4pt]
 & \multicolumn{1}{|c|}{\multirow{2}{*}{
\textbf{log(}$\mathbf{2}$\textbf{)}}} & \textbf{0.050} & \textbf{0.052} & \textbf{0.049} & \textbf{0.050} & \textbf{0.049} & \textbf{0.051} \\
 & \multicolumn{1}{|c|}{}  & \textbf{(0.008)} & \textbf{(0.000)} & \textbf{(0.050)} & \textbf{(0.047)} & \textbf{(0.048)} & \textbf{(0.046)} \\[4pt]
 & \multicolumn{1}{|c|}{\multirow{2}{*}{$\log(1.75)$}} & 0.076 & 0.077 & 0.157 & 0.136 & 0.506 & 0.427 \\
 & \multicolumn{1}{|c|}{}  & (0.017) & (0.000) & (0.155) & (0.135) & (0.501) & (0.424) \\[4pt]
 & \multicolumn{1}{|c|}{\multirow{2}{*}{$\log(1.5)$}} & 0.123 & 0.097 & 0.423 & 0.341 & 0.974 & 0.922 \\
 & \multicolumn{1}{|c|}{}  & (0.031) & (0.000) & (0.412) & (0.337) & (0.971) & (0.920) \\[4pt]
 & \multicolumn{1}{|c|}{\multirow{2}{*}{$\log(1.25)$}} & 0.176 & 0.120 & 0.751 & 0.612 & 1.000 & 0.999 \\
 & \multicolumn{1}{|c|}{}  & (0.042) & (0.000) & (0.747) & (0.606) & (1.000) & (0.999) \\[4pt]
\cline{2-8}
\cline{2-8}

\multicolumn{1}{c|}{\multirow{12}{*}{(High)}} 
 & \multicolumn{1}{|c|}{\multirow{2}{*}{$\log(2.5)$}} & \rule{0pt}{3ex} 0.035 & 0.046 & 0.010 & 0.012 & 0.000 & 0.000 \\
 & \multicolumn{1}{|c|}{}  & (0.000) & (0.000) & (0.009) & (0.009) & (0.000) & (0.000) \\[4pt]
 & \multicolumn{1}{|c|}{\multirow{2}{*}{$\log(2.25)$}} & 0.042 & 0.052 & 0.021 & 0.026 & 0.005 & 0.007 \\
 & \multicolumn{1}{|c|}{}  & (0.000) & (0.000) & (0.019) & (0.017) & (0.006) & (0.006) \\[4pt]
 & \multicolumn{1}{|c|}{\multirow{2}{*}{
\textbf{log(}$\mathbf{2}$\textbf{)}}} & \textbf{0.045} & \textbf{0.062} & \textbf{0.047} & \textbf{0.052} & \textbf{0.048} & \textbf{0.048} \\
 & \multicolumn{1}{|c|}{}  & \textbf{(0.000)} & \textbf{(0.000)} & \textbf{(0.043)} & \textbf{(0.030)} & \textbf{(0.050)} & \textbf{(0.052)} \\[4pt]
 & \multicolumn{1}{|c|}{\multirow{2}{*}{$\log(1.75)$}} & 0.054 & 0.069 & 0.100 & 0.088 & 0.280 & 0.230 \\
 & \multicolumn{1}{|c|}{}  & (0.000) & (0.000) & (0.095) & (0.048) & (0.277) & (0.226) \\[4pt]
 & \multicolumn{1}{|c|}{\multirow{2}{*}{$\log(1.5)$}} & 0.072 & 0.086 & 0.209 & 0.139 & 0.728 & 0.598 \\
 & \multicolumn{1}{|c|}{}  & (0.000) & (0.000) & (0.185) & (0.065) & (0.726) & (0.594) \\[4pt]
 & \multicolumn{1}{|c|}{\multirow{2}{*}{$\log(1.25)$}} & 0.076 & 0.092 & 0.327 & 0.179 & 0.981 & 0.911 \\
 & \multicolumn{1}{|c|}{}  & (0.000) & (0.000) & (0.313) & (0.071) & (0.977) & (0.909) \\[4pt]
\cline{2-8}
\end{tabular}
\caption{\it Simulated rejection probabilities of the  test \eqref{bootstrap1}  (blank) and TOST \eqref{TOST1} (brackets) for the subgroup-versus-subgroup hypotheses \eqref{hypothesis1} with significance level $\alpha = 0.05$ and consistency threshold $\theta_c = \log(2)$. The bold-font rows mark the null boundary. For each interaction coefficient $\beta_{TS}$ the dropout rate $\gamma$ is adjusted to keep a censoring proportion of 20 \% in table (Low), 50 \% in table (Medium) and 80 \% in table (High).}
\label{tab1}
\end{table}

\begin{table}[!htbp]
\footnotesize
\centering
\begin{tabular}{cc|c|c||c|c||c|c|}
\cline{3-8}
 & & \multicolumn{2}{c||}{$n=200$} & \multicolumn{2}{c||}{$n=800$} & \multicolumn{2}{c|}{$n=5000$} \\
\cline{2-8}
\multicolumn{1}{c|}{} &
\multicolumn{1}{c|}{$\beta_{TS}$} &
$p = 0.5$ & $p = 0.25$ &
$p = 0.5$ & $p = 0.25$ &
$p = 0.5$ & $p = 0.25$ \\
\cline{2-8}

\multicolumn{1}{c|}{\multirow{12}{*}{(Low)}} 
 & \multicolumn{1}{|c|}{\multirow{2}{*}{$\log(2.5)$}}  & \rule{0pt}{3ex} 0.007 & 0.017 & 0.001 & 0.001 & 0.000 & 0.000 \\
 & \multicolumn{1}{|c|}{} & (0.007) & (0.010) & (0.000) & (0.001) & (0.000) & (0.000) \\[4pt]
 & \multicolumn{1}{|c|}{\multirow{2}{*}{$\log(2.25)$}}  & 0.023 & 0.028 & 0.006 & 0.013 & 0.000 & 0.001 \\
 & \multicolumn{1}{|c|}{}  & (0.016) & (0.020) & (0.005) & (0.012) & (0.000) & (0.000) \\[4pt]
 & \multicolumn{1}{|c|}{\multirow{2}{*}{
\textbf{log(}$\mathbf{2}$\textbf{)}}} & \textbf{0.044} & \textbf{0.049} & \textbf{0.046} & \textbf{0.050} & \textbf{0.042} & \textbf{0.051} \\
 & \multicolumn{1}{|c|}{}  & \textbf{(0.034)} & \textbf{(0.033)} & \textbf{(0.037)} & \textbf{(0.050)} & \textbf{(0.042)} & \textbf{(0.048)} \\[4pt]
 & \multicolumn{1}{|c|}{\multirow{2}{*}{$\log(1.75)$}} & 0.091 & 0.101 & 0.188 & 0.183 & 0.657 & 0.555 \\
 & \multicolumn{1}{|c|}{}  & (0.077) & (0.058) & (0.175) & (0.183) & (0.621) & (0.548) \\[4pt]
 & \multicolumn{1}{|c|}{\multirow{2}{*}{$\log(1.5)$}} & 0.174 & 0.151 & 0.520 & 0.461 & 0.997 & 0.986 \\
 & \multicolumn{1}{|c|}{}  & (0.151) & (0.103) & (0.516) & (0.460) & (0.996) & (0.984) \\[4pt]
 & \multicolumn{1}{|c|}{\multirow{2}{*}{$\log(1.25)$}} & 0.295 & 0.230 & 0.885 & 0.831 & 1.000 & 1.000 \\
 & \multicolumn{1}{|c|}{}  & (0.260) & (0.163) & (0.879) & (0.822) & (1.000) & (1.000) \\[4pt]
\cline{2-8}
\cline{2-8}

\multicolumn{1}{c|}{\multirow{12}{*}{(Medium)}} 
 & \multicolumn{1}{|c|}{\multirow{2}{*}{$\log(2.5)$}} & \rule{0pt}{3ex} 0.015 & 0.023 & 0.004 & 0.005 & 0.000 & 0.000 \\
 & \multicolumn{1}{|c|}{}  & (0.002) & (0.003) & (0.002) & (0.005) & (0.000) & (0.000) \\[4pt]
 & \multicolumn{1}{|c|}{\multirow{2}{*}{$\log(2.25)$}} & 0.027 & 0.033 & 0.012 & 0.015 & 0.002 & 0.002 \\
 & \multicolumn{1}{|c|}{}  & (0.002) & (0.004) & (0.010) & (0.013) & (0.000) & (0.000) \\[4pt]
 & \multicolumn{1}{|c|}{\multirow{2}{*}{
\textbf{log(}$\mathbf{2}$\textbf{)}}} & \textbf{0.045} & \textbf{0.051} & \textbf{0.041} & \textbf{0.053} & \textbf{0.037} & \textbf{0.045} \\
 & \multicolumn{1}{|c|}{}  & \textbf{(0.005)} & \textbf{(0.005)} & \textbf{(0.044)} & \textbf{(0.051)} & \textbf{(0.043)} & \textbf{(0.051)} \\[4pt]
 & \multicolumn{1}{|c|}{\multirow{2}{*}{$\log(1.75)$}} & 0.075 & 0.084 & 0.144 & 0.153 & 0.470 & 0.426 \\
 & \multicolumn{1}{|c|}{}  & (0.007) & (0.006) & (0.144) & (0.143) & (0.462) & (0.421) \\[4pt]
 & \multicolumn{1}{|c|}{\multirow{2}{*}{$\log(1.5)$}} & 0.117 & 0.112 & 0.382 & 0.340 & 0.968 & 0.926 \\
 & \multicolumn{1}{|c|}{}  & (0.014) & (0.007) & (0.377) & (0.336) & (0.967) & (0.920) \\[4pt]
 & \multicolumn{1}{|c|}{\multirow{2}{*}{$\log(1.25)$}} & 0.156 & 0.125 & 0.735 & 0.616 & 1.000 & 1.000 \\
 & \multicolumn{1}{|c|}{}  & (0.018) & (0.008) & (0.727) & (0.614) & (1.000) & (0.999) \\[4pt]
\cline{2-8}
\cline{2-8}

\multicolumn{1}{c|}{\multirow{12}{*}{(High)}} 
 & \multicolumn{1}{|c|}{\multirow{2}{*}{$\log(2.5)$}} & \rule{0pt}{3ex} 0.031 & 0.026 & 0.009 & 0.012 & 0.000 & 0.001 \\
 & \multicolumn{1}{|c|}{}  & (0.000) & (0.000) & (0.008) & (0.008) & (0.000) & (0.000) \\[4pt]
 & \multicolumn{1}{|c|}{\multirow{2}{*}{$\log(2.25)$}} & 0.039 & 0.043 & 0.019 & 0.025 & 0.004 & 0.007 \\
 & \multicolumn{1}{|c|}{}  & (0.000) & (0.000) & (0.016) & (0.016) & (0.005) & (0.007) \\[4pt]
 & \multicolumn{1}{|c|}{\multirow{2}{*}{
\textbf{log(}$\mathbf{2}$\textbf{)}}} & \textbf{0.046} & \textbf{0.052} & \textbf{0.047} & \textbf{0.052} & \textbf{0.042} & \textbf{0.048} \\
 & \multicolumn{1}{|c|}{}  & \textbf{(0.000)} & \textbf{(0.000)} & \textbf{(0.040)} & \textbf{(0.034)} & \textbf{(0.045)} & \textbf{(0.047)} \\[4pt]
 & \multicolumn{1}{|c|}{\multirow{2}{*}{$\log(1.75)$}} & 0.055 & 0.052 & 0.101 & 0.092 & 0.262 & 0.235 \\
 & \multicolumn{1}{|c|}{}  & (0.000) & (0.000) & (0.088) & (0.053) & (0.258) & (0.235) \\[4pt]
 & \multicolumn{1}{|c|}{\multirow{2}{*}{$\log(1.5)$}} & 0.066 & 0.059 & 0.198 & 0.144 & 0.717 & 0.590 \\
 & \multicolumn{1}{|c|}{}  & (0.000) & (0.000) & (0.168) & (0.070) & (0.712) & (0.585) \\[4pt]
 & \multicolumn{1}{|c|}{\multirow{2}{*}{$\log(1.25)$}} & 0.076 & 0.073 & 0.308 & 0.183 & 0.977 & 0.913 \\
 & \multicolumn{1}{|c|}{}  & (0.000) & (0.000) & (0.290) & (0.081) & (0.973) & (0.910) \\[4pt]
\cline{2-8}
\end{tabular}
\caption{\it Simulated rejection probabilities of the  test \eqref{bootstrap3} (blank) and TOST-type test \eqref{FOST} (brackets) for the subgroups-versus-population hypotheses \eqref{weighted_hyp} with significance level $\alpha = 0.05$ and consistency threshold $\theta_c = \max\{p,1-p\} \log(2)$. The bold-font rows mark the null boundary. For each interaction coefficient $\beta_{TS}$ the dropout rate $\gamma$ is adjusted to keep a censoring proportion of 20 \% in table (Low), 50 \% in table (Medium) and 80 \% in table (High).}
\label{tab2}
\end{table}

\textbf{Results (Table \ref{tab1}).}
\textit{Type I error rate and power.} The simulation results confirm the asymptotic validity of the tests established in Theorem \ref{thm1} and Theorem \ref{thm2}: As the sample size $n$ grows the rejection probability of both tests converges to 0 for $|\beta_{TS}| > \theta_c$, $\alpha$ for $|\beta_{TS}| = \theta_c$ and 1 for $|\beta_{TS}| < \theta_c$, where $\theta_c = \log(2)$ and $\alpha = 0.05$. Overall, both tests control the Type I error rate in all three sample size regimes with the TOST having a slightly smaller Type I error rate than the  test \eqref{bootstrap1} in the interior of the null ($|\beta_{TS}| > \theta_c$). On the other hand, the test \eqref{bootstrap1} provides a, in some cases significantly, better approximation of the nominal level $\alpha$ on the null boundary ($|\beta_{TS}| = \theta_c$) and is also consistently, and in some cases substantially, more powerful under the alternative hypothesis ($|\beta_{TS}| < \theta_c$). For example, for $n=200$ and medium censoring, the rejection probabilities of the TOST on the null boundary are 0.008 and 0.000 in the balanced and unbalanced cases, respectively, whereas the corresponding rejection probabilities of the  test \eqref{bootstrap1} are $0.050$ and $0.052$. Considering the alternative $\beta_{TS} = \log(1.25)$ in the same censoring regime, the TOST only achieves a power of $0.042$ and $0.000$ for both cases, whereas the test \eqref{bootstrap1} provides a power of 0.176 and 0.120 for those cases. The superiority of the  test \eqref{bootstrap1} is remarkable in the small sample size scenario $n=200$ with medium to high censoring, becomes much less significant for the medium sample size $n=800$ and truly negligible in the large sample size scenario $n=5000$. The theoretical explanation for this observation is that the rejection probabilities of the two tests have identical limits which is why their difference vanishes as $n \to \infty$ (see Theorem \ref{thm1} and Theorem \ref{thm2}). \smallskip

\textit{The effect of subgroup sample size design.} Comparing the balanced ($n_0 = n_1$) and unbalanced ($n_0 \neq n_1$) case, we can note that while the Type I error rate of both tests does not differ significantly, the tests are usually less powerful in the unbalanced case compared to the balanced case. The loss in power can be notable for both tests, however the TOST seems to suffer more from an unbalanced design than the test \eqref{bootstrap1}. For example, for $n = 800$, high censoring (80 \%) and $\beta_{TS} = \log(1.5)$ the simulated power of the test \eqref{bootstrap1} and TOST is 0.209, respectively, 0.185 for $n_0 = n_1$ which reduces to 0.139, respectively, 0.065 when $n_0 \neq n_1$. In the same sub-table, for $\beta_{TS} = \log(1.25)$, the simulated power of the  test \eqref{bootstrap1} is $0.327$ which reduces to $0.179$ in the unbalanced case while the simulated power of the TOST is very similar, 0.313, but reduces to $0.071$ in the unbalanced case. Independent of the test in use, a balanced design is clearly preferable for consistency testing. \smallskip

\textit{The effect of censoring.} Censoring has very little effect on the Type I error rate of the test \eqref{bootstrap1}, in particular, the test still approximates the nominal level $\alpha$ well on the null boundary even in the high censoring and small sample size scenario. In contrast, the Type I error rate of the TOST is more sensitive to censoring and even reduces to zero on the null boundary for medium and high censoring in the small sample size scenario. Censoring has a major negative impact on the power of both tests and can decrease it considerably, however in the small sample size scenario the TOST is more affected by higher levels of censoring than the  test \eqref{bootstrap1}. In the medium and high censoring cases the TOST mostly has a power of zero under all three alternatives while the power of the  test \eqref{bootstrap1} still ranges from 7 \% to 18 \% in the medium censoring case and from 5 \% to 10 \% in the high censoring case. \smallskip

\textbf{Results (Table 2).} The overall findings are very similar to Table 1. In particular, the test \eqref{bootstrap3} is consistently more powerful than the test \eqref{FOST} based on the TOST principle and also significantly improves the approximation of the nominal level in the small sample size scenario. However, as indicated by Theorem \ref{thm4}, one disadvantage of the  test \eqref{bootstrap3} is that it tends to be conservative on the null boundary for $p= 1/2$, since the weight map $w(p) = \max \{p, 1- p \}$ is not differentiable at this point. Indeed, we can observe that even for the large sample size $n=5000$ the Type I error rate of the  test \eqref{bootstrap3} stays well below $\alpha = 0.05$ in all three censoring regimes on the null boundary.

\section{A CANTOS-motivated case study}
\label{sec5}

The practical question introduced in Section \ref{sec1} arising from the CANTOS trial provided a central motivation for the methodological developments in this paper (\cite{ridker2017antiinflammatory}). In this section, we return to this question and apply the developed consistency tests \eqref{bootstrap1} and \eqref{bootstrap3} to simulated data which closely resembles the real dataset of the CANTOS trial in order to assess the subgroup treatment effect consistency. \smallskip
 
 \textbf{CANTOS.}
 The Canakinumab Anti-Inflammatory Thrombosis Outcomes Study (CANTOS) was a randomized, double-blind, placebo-controlled Phase III trial lasting from April 2011 to June 2017 that studied the treatment effect of canakinumab on patients with a high risk for cardiovascular disease. Patients were randomized to placebo and three doses of canakinumab (50 mg, 150 mg and 300 mg) of which the 150 mg dose proved to be the most promising. The primary efficacy end point was the time to nonfatal myocardial infarction, any nonfatal stroke, or cardiovascular death measured in days after randomization to treatment. The study concluded that canakinumab is able to (significantly) reduce the rate of these cardiovascular events compared to placebo. The trial satisfied the proportional hazards assumption reasonably well and so fits into the context of this work. \smallskip

\textbf{Motivating question.} As the overall effect of canakinumab was positive, assessing the treatment effect consistency between trial subgroups and the population becomes a matter of interest, particularly for the 150 mg dose. Motivated by general regulatory interest in sex-specific treatment effects, we are interested in testing the consistency of treatment effects between women and the overall population. In other words, we aim to answer the question: \textit{Is the treatment effect in women sufficiently close to the treatment effect in the overall population?} If the answer is affirmative, this would strengthen the applicability of canakinumab to women and support the transfer of the overall trial finding to this subgroup. \smallskip 

\textbf{Simulated trial data.} The real dataset of the CANTOS trial is confidential, however our simulation design approximates the characteristics of the trial. We fix a sample size of $n= 5500$ patients which were randomized to placebo (3349) and treatment (2151) as well as male ($n_0 = 4103$) and female ($n_1 = 1397$) patients, where we assume a true female subgroup prevalence of $p=0.25$ (approximately matching the observed proportion in CANTOS). Accordingly, in the covariate vector $X_i = (X_{iT}, X_{iS})^\top$ corresponding to the $i$th patient, the binary covariates $X_{iT}$ and $X_{iS}$ decode the treatment group and sex, respectively. We assume all survival times in the proportional hazards model \eqref{hdn1} are exponentially distributed with (constant) baseline hazard function given by $h_0(t) = \exp(-9.3)$ and parameter vector $\beta = (-0.25, -0.2,\, \log(1.1))^\top$, where $\beta_T = -0.25$ represents the treatment effect of canakinumab, $\beta_S = -0.2$ the effect of sex and $\beta_{TS} = \log(1.1) \approx 0.10$ the treatment-by-sex interaction effect. Censoring times are generated iid from the random variable $C =\min(C_1, C_2)$, where the random variable $C_1 = \beta(1,1.5) \cdot 1000 + 1000$ represents censoring by study end and the random variable $C_2 = \text{Exp}(0.0000256)$ represents censoring by dropout, resulting in about 90 \% censoring. Here, $\beta(\cdot,\cdot)$ denotes a Beta-distributed random variable and $\text{Exp}(\cdot)$ denotes an exponential distributed random variable. \smallskip

\textbf{Consistency testing.} Note that we assume an interaction-effect exists, that is, sex has an influence on the treatment effect of canakinumab, however the effect $\beta_{TS}$ is small $-$ the two subgroup hazard ratios only differ by about 10 \%. In this case, a conventional interaction test would only be able to (correctly) reject the null hypothesis of no interaction, but this would yield no information about the \textit{size} of the interaction. However, judging the size of the interaction is exactly what is needed to answer the motivating question and the developed consistency tests can be used for this purpose. \smallskip

We start by testing the treatment effect consistency between men and women directly using Algorithm \ref{alg1}, as this is the strongest, and so most informative, type of consistency. We use the consistency margin $\theta_c = \log(1.6) \approx 0.47$ for the hypotheses \eqref{hypothesis1}. The estimated magnitude of the treatment-by-sex interaction coefficient is given by $|\hat \beta_{TS}| = 0.15$. The bootstrap quantile in \eqref{bootstrap1} for $\alpha = 0.05$ $(0.1)$ is $q^{*(1)}_\alpha = 0.12$ $(0.20)$, hence treatment effect consistency between men and women cannot be claimed for $\alpha = 0.05$, but for $\alpha = 0.1$.  The corresponding $p$-value  is given by $p_{alg} = 0.06$.  \medskip

Next, we assess the treatment effect consistency between men, women and the overall population simultaneously using Algorithm \ref{alg3} at the same consistency margin, which is the weaker statement. We estimate the unknown female subgroup prevalence $p$ by the proportion of female patients in the trial which is given by $\hat p = 0.254$, hence the estimated weighted absolute interaction is $\max(\hat p, 1 - \hat p) |\hat \beta_{TS}| = 0.11$. The bootstrap quantile in \eqref{bootstrap3} for $\alpha = 0.05$ $(0.1)$ is $q^{*(3)}_\alpha = 0.21$ $(0.27)$, hence treatment effect consistency between men, women and the overall population can be claimed at both significance levels. The corresponding $p$-value \eqref{pvalue_boot2} is given by $p_{boot} = 0.01$.

\section{Concluding remarks and future research}
\label{sec6}

In this paper, we developed a formal framework for assessing subgroup treatment-effect consistency for time-to-event outcomes within the Cox proportional hazards model. In contrast to conventional interaction testing, which examines whether subgroup-specific treatment effects are exactly equal, we formulate subgroup consistency as an equivalence problem. This formulation addresses the more clinically relevant question of whether differences between treatment effects are sufficiently small to be considered negligible according to a prespecified consistency margin.
\medskip

We considered two related inferential questions. The first concerns consistency of the treatment effects between two complementary subgroups. The second concerns consistency of one or both subgroup-specific treatment effects with the treatment effect in the overall population. The former represents the stronger requirement: consistency between the two subgroup effects implies consistency with the overall population effect at the same margin, whereas the converse does not generally hold. The appropriate formulation should therefore be determined by the scientific objective of the analysis.
\medskip

For each setting, we developed a conventional TOST-type procedure and a novel procedure motivated by optimal tests for consistency hypotheses in normally distributed data. The proposed procedure exploits the asymptotic normality of the partial likelihood estimator and uses critical values based on the corresponding folded normal distribution. Both procedures are asymptotically valid, but the proposed test is less conservative and more powerful than the TOST-type procedure under the conditions considered.
\medskip

The interpretation of the proposed analyses depends critically on the choice of the consistency margin. This margin should reflect the largest difference in treatment effects that would still be considered clinically acceptable and should ideally be prespecified in consultation with clinical experts. Establishing consistency does not imply that the treatment is efficacious within each subgroup, nor does it demonstrate that the treatment-by-subgroup interaction is exactly zero. Conversely, failure to establish consistency does not prove that clinically relevant heterogeneity is present; it may instead reflect limited information, particularly for small subgroups or when censoring is substantial.
\medskip

Several limitations suggest directions for further research. The proposed methods rely on the proportional hazards assumption and on independent, non-informative censoring. Extensions to settings with non-proportional hazards could be based on time-varying treatment effects or alternative estimands such as restricted mean survival time or milestone survival probabilities. The present work is also restricted to a single prespecified binary subgroup variable defining two complementary subgroups. Extensions to more than two subgroups, continuous effect modifiers, overlapping or cross-classified subgroups, and simultaneous evaluation of multiple subgroup variables would be valuable. Such extensions would also require appropriate consideration of multiplicity.
\medskip

A related direction is the development of methods for detecting clinically relevant heterogeneity rather than confirming consistency. Although the corresponding hypotheses can be obtained by reversing the roles of the null and alternative hypotheses in this work, the choice and interpretation of the clinical margin would need to be reconsidered. More generally, outside the proportional hazards framework, subgroup consistency may involve comparisons between entire time-varying treatment-effect functions rather than finite-dimensional parameters, requiring new asymptotic and resampling methods.

\bigskip

{\bf Software} At the GitHub repository \href{https://github.com/LKoletzko/Testing-for-subgroup-treatment-effect-consistency-in-the-Cox-model/tree/main}{github/LKoletzko}, we provide the R code that was used for the simulations and the case study as well as an R implementation of Algorithms \ref{alg1}-\ref{alg3} presented in this paper for direct practical use, together with detailed documentation.

\medskip

{\bf Acknowledgements} This  work has been funded by the Deutsche Forschungsgemeinschaft (DFG, German Research Foundation) — Project-ID 499552394 - SFB 1597 and  by the European Union through the European Joint Programme on Rare Diseases under the European Union's Horizon 2020 Research and Innovation Programme Grant Agreement Number 825575. HD is task leader of RealiseD supported by the Innovative Health Initiative Joint Undertaking (IHI JU) under grant agreement No. 101165912. The JU receives support from the European Union’s Horizon Europe research and innovation programme and COCIR, EFPIA, Europa Bío, MedTech Europe, and Vaccines Europe. Views and opinions expressed are those of the author(s) only. This publication reflects the author’s views. They do not necessarily reflect those of the Innovative Health Initiative Joint Undertaking and its members, who cannot be held responsible for them.

\bibliographystyle{chicago}
\bibliography{references}


\section{Supplementary Material}
\label{sec7}

In this supplement, we prove the asymptotic validity of all six tests provided in this work and also establish Assumption (P1) for the mle and outcome model-based estimator of the unknown subgroup prevalence $p$. Throughout this paper we assume that the regularity conditions for the asymptotic normality of the partial likelihood estimator $\hat \beta = ( \hat \beta_T, \hat \beta_S, \hat \beta_{TS})^\top$ are satisfied \citepSM[see, for example, conditions A-D in][]{andersen1982}. It then follows that 
    \begin{align} \label{p1}
      \mathcal{I}^{1/2}(\beta) \big ( \hat \beta - \beta \big ) \tod \mathcal{N}_3(0, I_3),
    \end{align} 
where $\mathcal{N}_3(0, I_3)$ denotes a $3$-dimensional standard normal distribution. Moreover, in the same reference it is shown that $ \frac{1}{n} \mathcal{I}(\beta) \topr \mathcal{I}_0(\beta)$ such that 
    \begin{align} \label{p1a}
      \sqrt{n}  \big ( \hat \beta - \beta \big ) \tod \mathcal{N}_3(0, \mathcal{I}^{-1}_0(\beta)).
    \end{align}

\subsection{Asymptotic validity of the tests \texorpdfstring{\eqref{TOST1}}{(TOST1)} and \texorpdfstring{\eqref{bootstrap1}}{(bootstrap1)}}

\begin{theorem}[Asymptotic validity of TOST \eqref{TOST1}]
\label{thm1}
{ \it 
    For all $\alpha \in (0,0.5)$ the TOST \eqref{TOST1} satisfies
    \begin{align*}
        \mathbb P\Big( \big[ \hat \beta_{TS} - z_{1- \alpha } \, \hat \sigma_{\hat \beta_{TS}}, \ 
\hat \beta_{TS} + z_{1- \alpha} \, \hat \sigma_{\hat \beta_{TS}} \big] \subseteq (- \theta_c, \theta_c) \Big) \ton \begin{cases}
            0 &, \quad |\beta_{TS}| > \theta_c \\
            \alpha &, \quad |\beta_{TS}| = \theta_c \\
            1 &, \quad |\beta_{TS}| < \theta_c.
        \end{cases}
    \end{align*}
}
\end{theorem}

\begin{proof}
First, note that  the decision rule \eqref{TOST1} is equivalent to 
$  | \hat \beta_{TS} | < \theta_c + z_{\alpha} \, \hat \sigma_{\hat \beta_{TS}}$,
and that the weak convergence \eqref{p1}, the continuous mapping theorem and Slutsky's Lemma yield
    \begin{align}
    \label{l1}
      \frac{\hat \beta_{TS} - \beta_{TS}}{\hat \sigma_{\hat \beta_{TS}}} \tod \mathcal{N}(0, 1).
    \end{align}
\textit{Case 1: $|\beta_{TS}| \neq \theta_c$.} Note that $|\hat \beta_{TS}| \topr |\beta_{TS}|$ by consistency of the partial mle. As $\hat \sigma_{\hat \beta_{TS}} \topr 0$, we have $|\hat \beta_{TS}| - z_{\alpha} \, \hat \sigma_{\hat \beta_{TS}} \topr |\beta_{TS}|$, which implies by Lemma 2.22 in \citeSM{Koletzko2024}
\begin{align*}
    \mathbb P\Big( |\hat \beta_{TS}| - z_{\alpha} \, \hat \sigma_{\hat \beta_{TS}}  < \theta_c \Big) \ton \begin{cases}
            0  &, \quad |\beta_{TS}| > \theta_c \\
            1  &, \quad |\beta_{TS}| < \theta_c.
        \end{cases}
\end{align*}

\textit{Case 2: $|\beta_{TS}| = \theta_c$.} 
    Note that $\beta_{TS} \neq 0$, since $\theta_c > 0$, and we assume without loss of generality  $\beta_{TS} >0 $.  Therefore (note that we can apply Freedman’s inequality to the martingale representation of the score to derive a concentration inequality for $ \hat \beta_{TS} -  \beta_{TS} $),
\begin{align*}
 \frac{|\hat \beta_{TS}| - \theta_c}{\hat \sigma_{\hat \beta_{TS}}} = 
      \sgn(\beta_{TS})
      \frac{\hat \beta_{TS} - \beta_{TS}}{\hat \sigma_{\hat \beta_{TS}}}   + o_{\mathbb{P}} (1) \tod \mathcal{N}(0, 1),
\end{align*}
which gives $
    \mathbb P\big(  | \hat \beta_{TS}| < \theta_c + z_{\alpha} \, \hat \sigma_{\hat \beta_{TS}} \big) \ton \alpha.
$
\end{proof}

Note that the proofs of Case 1 in all following Theorems rely on an application of Lemma 2.22 in \citeSM{Koletzko2024}, but we will refrain from mentioning it each time.

\begin{theorem}[Asymptotic validity of test \eqref{bootstrap1}]
\label{thm2}
{ \it
For all $\alpha \in (0,1)$ the  test \eqref{bootstrap1} satisfies
\begin{align*}
 \mathbb P\Big( |\hat \beta_{TS}| < q_\alpha^{*(1)} \Big) \ton \begin{cases}
            0 \ &, \quad |\beta_{TS}| > \theta_c \\
            \alpha \ &, \quad |\beta_{TS}| = \theta_c \\
            1 \ &, \quad |\beta_{TS}| < \theta_c.
        \end{cases}
\end{align*}
}
\end{theorem}

\begin{proof}

    \textit{Case 1: $|\beta_{TS}| \neq \theta_c$.} Note that \textit{(i)} $|\hat \beta_{TS}| \topr |\beta_{TS}|$
    and \textit{(ii)} $q_\alpha^{*(1)} \topr \theta_c$ so that $|\hat \beta_{TS}| - q_\alpha^{*(1)} \topr |\beta_{TS}| - \theta_c$ which immediately implies
\begin{align*}
    \mathbb P\Big( |\hat \beta_{TS}| - q_\alpha^{*(1)}  < 0 \Big) \ton \begin{cases}
            0  &, \quad |\beta_{TS}| - \theta_c > 0 \\
            1  &, \quad |\beta_{TS}| - \theta_c < 0.
        \end{cases}
\end{align*}
\textit{Proof of (i).} The convergence follows from the consistency of the partial mle and the continuous mapping theorem. \medskip

\textit{Proof of (ii).} Note that $\hat \sigma_{\hat \beta_{TS}}^2 \topr 0$ as a consequence of \eqref{l1} and  $\hat \beta_{TS} - \beta_{TS} \topr 0$. For $\beta_{TS}^* \sim \mathcal{N}(\theta_c, \hat \sigma_{\hat \beta_{TS}}^2)$, this implies 
\begin{align}
\label{l11}
    \beta_{TS}^* - \theta_c \topr 0 \quad \text{conditionally in probability},
\end{align}
since for any $\varepsilon > 0$
\begin{align*}
    \mathbb P\Big( |\beta_{TS}^* - \theta_c| > \varepsilon \ | \ \mathcal{D}_n \Big) \leq \frac{\mathbb E\big( (\beta_{TS}^* - \theta_c)^2  \mid \mathcal{D}_n \big)}{\varepsilon^2} = \frac{\hat \sigma_{\hat \beta_{TS}}^2}{\varepsilon^2} \topr 0
\end{align*}
by the conditional Markov inequality and since $\beta_{TS}^* - \theta_c$ is mean-zero with variance $\hat \sigma_{\hat \beta_{TS}}^2$ conditional on the data $\mathcal{D}_n$ by definition of $\beta_{TS}^*$. Because the absolute value is continuous on $\R$, the convergence \eqref{l11} then implies
\begin{align*}
    |\beta_{TS}^*| - \theta_c \topr 0 \quad \text{conditionally in probability},
\end{align*}
which finally implies for the corresponding conditional $\alpha$-quantile sequence 
\begin{align*}
    q_\alpha^{*(1)} - \theta_c \topr 0 \quad \forall \alpha \in (0,1).
\end{align*}

\textit{Case 2: $|\beta_{TS}| = \theta_c$.} By the same reasoning as in the proof of Theorem \ref{thm1}, we obtain
\begin{align}
\label{l2}
      \frac{|\hat \beta_{TS}| - \theta_c}{\hat \sigma_{\hat \beta_{TS}}} \tod \mathcal{N}(0, 1).
\end{align}
Let $\beta_{TS}^* \sim  \mathcal{N} \big( \theta_c , \ \hat \sigma_{\hat \beta_{TS}}^2 \big)$ and note that the absolute value $|\cdot|$ is differentiable at $\theta_c > 0$ with derivative $|\cdot|^{'}_{\theta_c}(x) = x$ and so the bootstrap delta method (Theorem 23.5, \citeSM{vanderVaart1998}) together with a bootstrap version of Slutsky's Lemma imply
\begin{align*}
   \frac{|\beta_{TS}^*| - \theta_c}{\hat \sigma_{\hat \beta_{TS}}}  \tod \mathcal{N}(0,1) \quad \text{conditionally in probability,}
\end{align*}
hence the corresponding conditional $\alpha$-quantile sequence satisfies
\begin{align}
\label{l3}
    \frac{q_\alpha^{*(1)} - \theta_c}{\hat \sigma_{\hat \beta_{TS}}}  \topr z_\alpha \qquad \forall \alpha \in (0,1).
\end{align}
Now, the two convergences \eqref{l2} and \eqref{l3} together imply
\begin{align*}
    \mathbb P \Big( |\hat \beta_{TS}| < q_\alpha^{*(1)} \Big) \ton \alpha \qquad \forall \alpha \in (0,1).
\end{align*}

\end{proof}

\subsection{Asymptotic validity of the tests \texorpdfstring{\eqref{TOST2}}{(TOST2)}, \texorpdfstring{\eqref{bootstrap2}}{(bootstrap2)}, \texorpdfstring{\eqref{FOST}}{(FOST)} and \texorpdfstring{\eqref{bootstrap3}}{(bootstrap3)}}

Throughout this section we assume 
that the estimator $\hat p$ for $p$ satisfies
\begin{align}
\label{p1b}
\sqrt{n} \big( (\hat \beta_{TS}, \hat p)^\top - (\beta_{TS}, p)^\top \big) \tod \mathcal{N}_2(0, V)
\end{align}
and $V$ can be consistently estimated from the data by some $\hat V$. Consequently, for large sample sizes $n$, the vector $(\hat \beta_{TS}, \hat p)^\top$ is approximately normal distributed  and the variance $\Sigma$ can be approximated by $\hat \Sigma := \hat V/n$.

\begin{theorem}[Asymptotic validity of TOST \eqref{TOST2}]
\label{thm5}
{ \it If $w$ is continuously differentiable at $p$, the TOST \eqref{TOST2} satisfies for all $\alpha \in (0,0.5)$
\begin{align*}
        \mathbb P\Big( \hat I_n^{w(p)} 
    \subseteq (-\theta_c, \theta_c) \Big) \ton \begin{cases}
            0 &, \quad w(p)|\beta_{TS}| > \theta_c \\
            \alpha &, \quad w(p)|\beta_{TS}| = \theta_c \\
            1 &, \quad w(p)|\beta_{TS}| < \theta_c.
        \end{cases}
\end{align*}
}
\end{theorem}

\begin{proof} First, note that the decision rule  \eqref{TOST2} is equivalent to $w(\hat p)| \hat \beta_{TS} | < \theta_c + z_{\alpha} \hat \sigma_{w(p)}$. \medskip

\textit{Case 1: $w(p) |\beta_{TS}| \neq \theta_c$.} Note that $w(\hat p)|\hat \beta_{TS}| \topr w(p)|\beta_{TS}|$ by \eqref{p1b} and the continuous mapping theorem and $ \hat \sigma_{w(p)} \topr 0$. Therefore, $w(\hat p)|\hat \beta_{TS}| - z_{\alpha} \hat \sigma_{w(p)} \topr w(p)|\beta_{TS}|$ which immediately implies
\begin{align*}
    \mathbb P\Big( w(\hat p)|\hat \beta_{TS}| - z_{\alpha} \hat \sigma_{w(p)} < \theta_c \Big) \ton \begin{cases}
            0  &, \quad w(p)|\beta_{TS}| > \theta_c \\
            1  &, \quad w(p)|\beta_{TS}| < \theta_c.
        \end{cases}
\end{align*}

\textit{Case 2: $w(p)|\beta_{TS}| = \theta_c$.} Note that $\beta_{TS} \neq 0$, since $\theta_c > 0$, hence the absolute value $|\cdot|$ is differentiable at $\beta_{TS}$ and $w$ is differentiable at $p$ by assumption, which implies that the map
\begin{align*}
    f: \R^2 \to \R, \quad f(x,y) = |x| w(y)
\end{align*}
is differentiable at $(\beta_{TS},p)^\top$ with derivative given by
\begin{align*}
    f_{(\beta_{TS},p)}^{'}(a,b) = (\sgn(\beta_{TS}) w(p), |\beta_{TS}| w^{'}(p)) \cdot (a,b)^\top.
\end{align*}
Therefore, the delta method (Theorem 3.1, \citeSM{vanderVaart1998}) applied to \eqref{p1b} yields (note $w(p)|\beta_{TS}| = \theta_c$)
\begin{align*}
      \frac{w(\hat p)|\hat \beta_{TS}| - \theta_c}{ \hat \sigma_{w(p)}} \tod \mathcal{N}(0, 1),
\end{align*}
which implies  $
    \mathbb P\big(  w(\hat p)| \hat \beta_{TS}| < \theta_c + z_{\alpha} \hat \sigma_{w(p)} \big) \ton \alpha.$
\end{proof}

\begin{theorem}[Asymptotic validity of test \eqref{bootstrap2}]
\label{thm6}
    { \it
If $w$ is continuously differentiable at $p$, then the  test \eqref{bootstrap2} satisfies for all $\alpha \in (0,1)$
\begin{align*}
 \mathbb P\Big( w(\hat p) \, |\hat \beta_{TS}|  < q_\alpha^{*(2)} \Big) \ton \begin{cases}
            0 \ &, \quad w(p)|\beta_{TS}| > \theta_c \\
            \alpha \ &, \quad w(p)|\beta_{TS}| = \theta_c \\
            1 \ &, \quad w(p)|\beta_{TS}| < \theta_c.
        \end{cases}
\end{align*}
}
\end{theorem}

\begin{proof}
    
\textit{Case 1: $w(p) |\beta_{TS}| \neq \theta_c$.} Note that $w(\hat p) |\hat \beta_{TS}| \topr w(p) |\beta_{TS}|$ by \eqref{p1b} and the continuous mapping theorem and $q_\alpha^{*(2)} \topr \theta_c$, since $\hat \sigma^2_{w(p)} \topr 0$, so that $w(\hat p) |\hat \beta_{TS}| - q_\alpha^{*(2)} \topr w(p) |\beta_{TS}| - \theta_c$, which immediately implies
\begin{align*}
    \mathbb P\Big( w(\hat p) |\hat \beta_{TS}| - q_\alpha^{*(2)} < 0 \Big) \ton \begin{cases}
        0 \ &, \quad w(p) |\beta_{TS}| - \theta_c > 0 \\
        1 \ &, \quad w(p) |\beta_{TS}| - \theta_c < 0.
\end{cases}
\end{align*}

\textit{Case 2: $w(p) |\beta_{TS}| = \theta_c$.} By the same arguments as given for the proof of Case 2 in Theorem \ref{thm5} we obtain 
\begin{align}
\label{l16}
      \frac{w(\hat p)|\hat \beta_{TS}| - \theta_c}{ \hat \sigma_{w(p)}} \tod \mathcal{N}(0, 1).
\end{align}
By similar arguments as given for the proof of Case 2 in Theorem \ref{thm2} we obtain
\begin{align}
\label{l21}
    \frac{q_\alpha^{*(2)} - \theta_c}{\hat \sigma_{w(p)}}  \topr z_\alpha \qquad \forall \alpha \in (0,1).
\end{align}
Now, the two convergences \eqref{l16} and \eqref{l21} imply the claim.
\end{proof}

\begin{theorem}[Asymptotic validity of test \eqref{FOST}]
\label{thm3}
{ \it The test \eqref{FOST} for the hypotheses \eqref{weighted_hyp} satisfies for all $\alpha \in (0,0.5)$
\begin{align*}
    \mathbb P\big(  ~ \hat I_n^{p}   \cup    \hat I_n^{1-p} \subseteq (-\theta_c,\theta_c) ~ \big) \ton \begin{cases}
            0 \ &, \quad \max \{ p, 1 - p \}  |\beta_{TS}| > \theta_c \\
             \alpha \ &, \quad \max \{ p, 1 - p \} |\beta_{TS}| = \theta_c \\
            1 \ &, \quad \max \{ p, 1 - p \} |\beta_{TS}| < \theta_c.
        \end{cases}
\end{align*}
}
\end{theorem}

\begin{proof}
    First, note that the decision rule \eqref{FOST} is equivalent to the condition
    \begin{align*}
    Z_n :=\max \{ |\hat p \hat \beta_{TS}| -  z_\alpha \hat \sigma_{p}, \ |(1-\hat p) \hat \beta_{TS}| - z_\alpha \hat \sigma_{1-p} \} < \theta_c.
    \end{align*}

\textit{Case 1: $\max \{ p, 1 - p \} |\beta_{TS}| \neq \theta_c$.} Since $\hat p \topr p$, $\hat \beta_{TS} \topr \beta_{TS}$ and $\hat \sigma_{p}, \ \hat \sigma_{1-p} \topr 0 $, the continuous mapping theorem yields 
\begin{align}
\label{l7}
    Z_n \topr \max \{ p, 1 - p \} |\beta_{TS}|.
\end{align}
Now, the convergence \eqref{l7} immediately implies
\begin{align*}
    \mathbb P(  Z_n < \theta_c ) \ton \begin{cases}
            0 \ &, \quad \max \{ p, 1 - p \} |\beta_{TS}| > \theta_c \\
            1 \ &, \quad \max \{ p, 1 - p \} |\beta_{TS}| < \theta_c.
    \end{cases}
\end{align*}

\textit{Case 2: $\max \{ p, 1 - p \} |\beta_{TS}| = \theta_c$.} We distinguish the two cases $p \neq 1/2$ and $p = 1/2$. First, let $p \neq 1/2$ and further assume $p<1/2$, the case $p>1/2$ follows by analogous arguments. Since $p<1/2$, we have $ |p\beta_{TS}| < \theta_c$ and so
\begin{align*}
    \mathbb P\Big( |\hat p \hat \beta_{TS}| -  z_\alpha \hat \sigma_{p} < \theta_c \Big) \ton 1,
\end{align*}
which implies
\begin{align*}
   \limsup_{n \to \infty} \ \mathbb P(  Z_n < \theta_c ) = \limsup_{n \to \infty} \  \mathbb P\Big( |(1-\hat p) \hat \beta_{TS}| - z_\alpha \hat \sigma_{1-p} < \theta_c \Big).
\end{align*}
Since $|(1-p) \beta_{TS}| = \theta_c > 0$ by assumption, we have $\beta_{TS} \neq 0$ and so the delta method (Theorem 3.1, \citeSM{vanderVaart1998}) applied to the function $f$ in the proof of Theorem \ref{thm5} with $w(y) = 1-y$ yields (note $|(1-p)\beta_{TS}| = \theta_c$) 
\begin{align*}
     \frac{ |(1-\hat p) \hat \beta_{TS}| - \theta_c }{\hat \sigma_{1-p}} \tod \mathcal{N}(0, 1),
\end{align*}
which immediately ensures
\begin{align*}
    \limsup_{n \to \infty} \  \mathbb P\Big(   |(1- \hat p) \hat \beta_{TS}| - z_\alpha \hat \sigma_{1-p} < \theta_c \Big) = \alpha.
\end{align*}
Next, let $p=1/2$. Define the sequences of sets
\begin{align*}
    A_n := \Big\{
\underbrace{\frac{|\hat p \hat \beta_{TS}| - \theta_c}{\hat \sigma_{p}}   }_{=: X_n} < z_\alpha \Big\} \quad , \quad B_n = \Big\{ \underbrace{ \frac{ |(1-\hat p) \hat \beta_{TS}| - \theta_c}{\hat \sigma_{1-p}}    }_{=: Y_n} < z_\alpha \Big\}
\end{align*}
and note that $\mathbb P(A_n) \ton \alpha$, since $X_n \tod \mathcal{N}(0,1)$ by the delta method. Therefore, to obtain $\mathbb P(A_n \cap B_n) \ton \alpha$, it suffices to show that
\begin{align*}
   (\star) \quad \mathbb P(A_n \backslash B_n) \ton 0.
\end{align*}
To prove $(\star)$, note that $X_n - Y_n \topr 0$, since $p = 1/2$. Now, let $\epsilon > 0$ and observe
\begin{align*}
  \mathbb P( X_n < z_\alpha , \ Y_n \geq z_\alpha ) &= \mathbb P( X_n < z_\alpha , \ Y_n \geq z_\alpha , \ |Y_n - z_\alpha| > \epsilon ) + \mathbb P( X_n < z_\alpha , \ Y_n \geq z_\alpha , \ |Y_n - z_\alpha| \leq \epsilon ) \\
  &\leq \mathbb P( | X_n - Y_n| > \epsilon) + \mathbb P(| Y_n - z_\alpha| \leq \epsilon),
\end{align*}
where the first summand in the upper bound converges to zero, since $X_n - Y_n \topr 0$ and the second summand converges to  $\mathbb P(|Z-z_\alpha| < \epsilon), \quad Z \sim \mathcal{N}(0,1)$ which can be made arbitrarily small by letting $\epsilon \downarrow 0$, since the distribution function of the standard normal is continuous at $z_\alpha$. Therefore, $(\star)$ must be true.
\end{proof}

In Theorem \ref{thm4} below we prove the asymptotic validity of the test \eqref{bootstrap3} for a general class of weight functions which includes the specific choice $w(p)=\max\{p,1-p\}$. The corresponding more general version of Algorithm \ref{alg3} that Theorem \ref{thm4} relates to is obtained by replacing all $\max(\cdot, 1-\cdot)$ by $w(\cdot)$. \smallskip

\begin{theorem}[Asymptotic validity of test \eqref{bootstrap3}]
\label{thm4}
{ \it
Let $w: \R \to (0,1)$ be Hadamard directionally differentiable with a sub-additive derivative and let $\mathcal{E} \subseteq \mathbb R$ denote the set of differentiability points, then for all $\alpha \in (0,1)$ the test \eqref{bootstrap3} for the weighted hypotheses \eqref{weighted_hyp1} satisfies
\begin{align*}
&(i) \ \lim_{n \to \infty} \ \mathbb P\Big( w(\hat p) |\hat \beta_{TS}| < q_\alpha^{*(3)} \Big) = \begin{cases}
             0 \ &, \quad w(p) \ |\beta_{TS}| > \theta_c \\
            1 \ &, \quad w(p) \ |\beta_{TS}| < \theta_c.
        \end{cases} \\[8pt]
 &(ii) \ \limsup_{n \to \infty} \ \mathbb P\Big( w(\hat p) |\hat \beta_{TS}| < q_\alpha^{*(3)} \Big)  \begin{cases}
            \leq \alpha \ &, \quad p \notin \mathcal{E} \\
            = \alpha \ &, \quad p \in \mathcal{E}.
        \end{cases}, \quad \text{if} \quad  w(p) |\beta_{TS}| = \theta_c.
\end{align*}
}
\end{theorem}

\begin{proof}

\textit{Proof of (i).} Assume $w(p) |\beta_{TS}| \neq \theta_c$. Note that \textit{(i.i)} $w(\hat p) |\hat \beta_{TS}| \topr w(p) |\beta_{TS}|$ and \textit{(i.ii)} $q_\alpha^{*(3)} \topr \theta_c$ so that $w(\hat p) |\hat \beta_{TS}| - q_\alpha^{*(3)} \topr w(p) |\beta_{TS}| - \theta_c$, which immediately implies
\begin{align*}
    \mathbb P\Big( w(\hat p) |\hat \beta_{TS}| - q_\alpha^{*(3)} < 0 \Big) \ton \begin{cases}
        0 \ &, \quad w(p) |\beta_{TS}| - \theta_c > 0 \\
        1 \ &, \quad w(p) |\beta_{TS}| - \theta_c < 0.
\end{cases}
\end{align*}

\textit{Proof of (i.i).} The convergence follows from the continuous mapping theorem, since $\hat p$ and $\hat \beta_{TS}$ are consistent and the map $w$ is continuous at $p$ as a consequence of being Hadamard directionally differentiable at $p$. \medskip

\textit{Proof of (i.ii).} Note that $\hat \Sigma := \hat V/n \topr 0$, since $\hat V \topr V$ by \eqref{p1b}. This implies 
\begin{align}
\label{l12}
        \big( \hat \beta_{TS}^*, \hat p^* \big)^\top - \big( \theta_c / w(\hat p), \hat p \big)^\top \topr 0 \quad \text{conditionally in probability}
\end{align}
by the conditional Markov inequality and definition of $(\hat \beta_{TS}^*, \hat p^*)^\top$. Because the function $h$ in Lemma \ref{lem1} below is continuous on $\R^2$, the convergence \eqref{l12} then implies (note $h( \theta_c/w(\hat p) , \hat p  ) = \theta_c$)
\begin{align*}
    w(\hat p^*) |\hat \beta_{TS}^*| - \theta_c \topr 0 \quad \text{conditionally in probability} 
\end{align*}
which finally implies for the corresponding conditional $\alpha$-quantile sequence 
\begin{align*}
    q_\alpha^{*(3)} - \theta_c \topr 0 \quad \forall \alpha \in (0,1).
\end{align*}

\textit{Proof of (ii).} Assume $w(p) |\beta_{TS}| = \theta_c$. First, note that $\beta_{TS} \neq 0$ in this case, since $\theta_c > 0$. Next, we establish a technical Lemma in order to derive the asymptotic error distribution of $w(\hat p) |\hat \beta_{TS}|$ via the delta method.

\begin{lemma}
\label{lem1}
{ \it
If $w : \R \to (0,1)$ is Hadamard directionally differentiable and $\mathcal{E}$ denotes the set of differentiability points, then the map
\begin{align*}
    h: \mathbb R^2 \to \mathbb R, \quad h(x,y) = |x| w(y)
\end{align*}
is Hadamard directionally differentiable at any $(x,y)^\top \in \mathbb R^2, \ x \neq 0$ with derivative map
\begin{align*}
    h_{(x,y)}^{'}(u,v) = \begin{cases}
        (\sgn(x)w(y), |x| w^{'}(y) ) \cdot (u,v)^\top &, \quad y \in \mathcal{E} \\
       w(y) \sgn(x)u + w_y^{'}(v)|x|   &, \quad y \notin \mathcal{E}.
    \end{cases}
\end{align*}
}
\end{lemma}

\begin{proof}
    For $y \in \mathcal{E}$ the map $h$ is differentiable at $(x,y)^\top , x \neq 0$ and so the conclusion follows from computing the partial derivatives. For $y \notin \mathcal{E}$, note that the maps $f(x,y) = w(y)$ and $g(x,y)=|x|$ are Hadamard directionally differentiable at $(x,y)^\top$ with derivatives $f_{(x,y)}^{'}(u,v)= w_y^{'}(v)$ and $g_{(x,y)}^{'}(u,v) = \sgn(x)u$, hence by the product rule (see \citeSM{shapiro1990concepts}) $h$ is Hadamard directionally differentiable at $(x,y)^\top$ with derivative
    \begin{align*}
        h_{(x,y)}^{'}(u,v) &= f(x,y) \cdot g_{(x,y)}^{'}(u,v) + f_{(x,y)}^{'}(u,v) \cdot g(x,y)  \\
        &= w(y) \sgn(x)u + w_y^{'}(v)|x|.
    \end{align*}
\end{proof}

With the help of Lemma \ref{lem1} we can apply the delta method for Hadamard directionally differentiable functions (Proposition 2.1, \citeSM{carcamo2020directional}) to the convergence \eqref{p1b} and obtain (note $w(p) |\beta_{TS}| = \theta_c$ by assumption)
\begin{align}
\label{l9}
    \sqrt{n} \Big( w(\hat p) |\hat \beta_{TS}| -\theta_c \Big) \tod \begin{cases}
        \mathcal{N}(0, v_{w(p)}^2)  &, \quad p \in \mathcal{E} \\
        w(p) \sgn(\beta_{TS}) Z_1 + w_p^{'}(Z_2)|\beta_{TS}| &, \quad p \notin \mathcal{E},
    \end{cases}
\end{align}
where $(Z_1, Z_2)^\top \sim \mathcal{N}_2(0, V)$ and $v_{w(p)}^2 := (w(p), w^{'}(p) \beta_{TS}) \ V \ (w(p), w^{'}(p) \beta_{TS})^\top.$ We now distinguish the two cases $p \in \mathcal{E}$ and $p \notin \mathcal{E}$. \medskip

\textit{Case 1: $p \in \mathcal{E}$.}
By definition of $(\hat \beta_{TS}^*, \hat p^*)^\top$, we have (conditional on the data)
\begin{align}
\label{l8}
        (\hat \beta_{TS}^*, \hat p^*)^\top \sim \mathcal{N}_2\Big( \, \big( \theta_c/w(\hat p), \hat p \big)^\top, \ \hat \Sigma \, \Big)
\end{align}
and since the map $h$ in Lemma \ref{lem1} is differentiable at any $(x,p)^\top$ for $p \in \mathcal{E}$ and $x \neq 0$, the bootstrap delta method (Theorem 23.5, \citeSM{vanderVaart1998}) applied to \eqref{l8} together with a bootstrap version of Slutsky's Lemma yields (note $h( \theta_c/w(\hat p) , \hat p  ) = \theta_c$)
\begin{align*}
   \frac{w(\hat p^*) |\hat \beta_{TS}^*| -\theta_c}{\hat \sigma_{w(p)}} \tod \mathcal{N}(0, 1),
\end{align*}
where $\hat \sigma_{w(p)}^2$ is defined in \eqref{varestim}.  Therefore, the corresponding conditional $\alpha$-quantile sequence satisfies
\begin{align}
\label{l10}
    \frac{q_\alpha^{*(3)} -\theta_c }{\hat \sigma_{w(p)}}  \topr z_\alpha \quad \forall \alpha \in (0,1).
\end{align}
Now, the two convergences \eqref{l9} and \eqref{l10} together imply
\begin{align*}
  \mathbb P\Big( w(\hat p) |\hat \beta_{TS}| < q_\alpha^{*(3)} \Big) \ton \alpha \quad \forall \alpha \in (0,1).
\end{align*}

\textit{Case 2: $p \notin \mathcal{E}$.} Note that by Lemma \ref{lem1} the function $h$ is not differentiable at $(\beta_{TS},p)^\top$ in this case and so the bootstrap delta method fails (see \citeSM{fang2019inference} for further details). However, since the Hadamard directional derivative $w_p^{'}(\cdot)$ is sub-additive by assumption, the Hadamard directional derivative $h_{(\beta_{TS},p)}^{'}(\cdot)$ also remains sub-additive, so that similar arguments as given for the proof of Theorem 2.1 (a) in \citeSM{BASTIAN2024} yield 
\begin{align*}
    \limsup_{n \to \infty} \ \mathbb P\Big( w(\hat p) |\hat \beta_{TS}| < q_\alpha^{*(3)} \Big) \leq \alpha.
\end{align*}
\end{proof}

\subsection{Estimating the subgroup prevalence} \label{sec83}

\begin{proposition}
\label{prop0}
    {\it The  trial subgroup proportion $\hat p_{mle}$ satisfies \eqref{p1b}. Moreover, 
    the asymptotic covariance matrix $V$  can be consistently estimated by
\begin{align}
\label{covest}
    \hat V = \frac{1}{n} \sum_{i=1}^{n} \hat \psi_i \hat \psi_i^\top,
\end{align}
where
\begin{align*}
    \hat \psi_i = \begin{pmatrix}
        e_3^\top \, ( \frac{1}{n}  \, \mathcal{I}(\hat \beta)  )^{-1} \,  \delta_i\bigl(\tilde X_i-\bar X(\hat \beta,T_{\mathrm{obs},i})\bigr) \\
        X_{iS}- \hat p_{mle}
    \end{pmatrix},
\end{align*}
$\mathcal{I}(\hat \beta)$ is  the plug-in estimator of the observed information matrix defined in \eqref{fisher_est} and  $\bar X(\beta, T_{{\rm obs},i})$ denotes the \textit{risk-set weighted mean}  defined by
\begin{align}
\label{riskmean} 
\bar X(\beta, T_{{\rm obs},i})
:= \frac{\sum_{j\in R(T_{{\rm obs},i})} \tilde X_j \exp(\beta^\top \tilde X_j)}{\sum_{j\in R(T_{{\rm obs},i})} \exp(\beta^\top \tilde X_j)}
\end{align}
(note that the  element $\hat V_{1,1}$  of the matrix $\hat V$  is given by $ e_3^\top ( \frac{1}{n}  \, \mathcal{I}(\hat \beta)  )^{-1} e_3 $ ).} 
\end{proposition}

\begin{proof}
It follows from the asymptotic normality in \eqref{p1a} and its proof in \citeSM{andersen1982} that
\begin{align}
\label{res}
    \sqrt{n} \Big( \begin{pmatrix}
        \hat \beta_{TS} \\
        \hat p_{mle}
    \end{pmatrix} 
    -
    \begin{pmatrix}
        \beta_{TS} \\
        p
    \end{pmatrix} \Big)
    \tod \mathcal{N}_2(0, V),
\end{align}
where $V = \text{Var}(\psi_i) = \mathbb E(\psi_i \psi_i^\top)$,
$$
 \psi_i:=  \begin{pmatrix}
        e_3^\top \,  \mathcal{I}_0(\beta)^{-1} \,  \delta_i \big ( \tilde X_i - \bar X(\beta, T_{{\rm obs},i}) \big ) \\
        X_{iS}- p  
    \end{pmatrix}, 
$$
and  $\bar X(\beta, T_{{\rm obs},i})$  is defined in \eqref{riskmean}. Therefore, the asymptotic covariance matrix $V$ in \eqref{res} can be consistently estimated by \eqref{covest}.
\end{proof}

We can derive a similar result for the outcome model-based estimator \eqref{outcome}. For this purpose we make the following assumptions
    \begin{enumerate}[leftmargin=4em, labelindent=0pt]
    \item[(OM1)] The trial and non-trial samples are independent.
    \item[(OM2)] The trial and non-trial sample sizes $n \to \infty$ and $n^* \to \infty$ such that $n/n^* \to \kappa \in (0,\infty)$.
    \item[(OM3)] The conditional probabilities $p(Z_i) := \mathbb P(X_{iS} =1 | Z_i)$ and $p(Z_i^*) := \mathbb P(X_{iS}^* =1 | Z_i^*)$ in \eqref{itexp} are given by a logistic regression model, that is
\begin{align}
    \label{model1}
        p(Z_i) = \text{expit}(\eta^\top Z_i), \quad p(Z_i^*) = \text{expit}(\eta^\top Z_i^*).
\end{align}
    Note that only the $X_{iS}$ are observed while the $X_{iS}^*$ are unobserved.
\end{enumerate}
Note that condition (OM3) can be replaced by other parametric model assumptions for the conditional probability.

\begin{proposition} \label{prob0a}
        { \it If Assumptions (OM1)--(OM3) hold, 
the outcome model-based estimator $\hat p_{OM}$ in \eqref{outcome} satisfies  \eqref{p1b}. Moreover,  
    the asymptotic covariance matrix $V$  can be consistently estimated by
    \begin{align}
    \label{vestim}
    \hat V = \frac{1}{n} \sum_{i=1}^{n} \hat \psi_i^{(1)} \hat \psi_i^{(1)\top} + \frac{n}{n^*} \frac{1}{n^*} \sum_{i=1}^{n^*} \hat \psi_i^{(2)} \hat \psi_i^{(2)\top},
\end{align}
where 
\begin{align*}
    \hat \psi_i^{(1)} &= \begin{pmatrix}
        e_3^\top \, ( \frac{1}{n}  \, \mathcal{I}(\hat \beta)  )^{-1} \, \delta_i\bigl(\tilde X_i-\bar X(\hat \beta,T_{\mathrm{obs},i})\bigr) \\
        \hat \Gamma^\top \hat \phi_i^{\eta}
    \end{pmatrix},~~~~
     \hat \psi_i^{(2)} = \begin{pmatrix}
         0 \\
        \operatorname{expit}(\hat \eta^\top Z_i^*) - \hat p_{OM}
          \end{pmatrix}~, \\
    \hat \phi_i^{\eta} & = \hat I_\eta^{-1} s_i(\hat \eta), \quad \hat I_\eta = \frac{1}{n} \sum_{i=1}^n \operatorname{expit}(\hat \eta^\top Z_i)(1-\operatorname{expit}(\hat \eta^\top Z_i)) Z_i Z_i^\top \\ 
   \hat \Gamma  & = \frac{1}{n^*} \sum_{i=1}^{n^*} \big(\operatorname{expit}(\hat \eta^\top Z_i^*)(1-\operatorname{expit}(\hat \eta^\top Z_i^*) \big) Z_i^* 
\end{align*}
and 
$\mathcal{I}(\hat \beta)$ is defined in \eqref{fisher_est}.
    }
\end{proposition}

\begin{proof} The Bahadur expansion of  the   mle $\hat \eta$ in the logistic regression model is given by 
    \begin{align*}
        \sqrt{n}( \hat \eta - \eta) = \frac{1}{\sqrt{n}} \sum_{i=1}^n \phi_i^{\eta} + o_{\mathbb P}(1),
    \end{align*}
    where $\phi_i^\eta = I_{\eta}^{-1} \, s_i(\eta)$, the matrix $I_\eta = \mathbb E_Z[\text{expit}(\eta^\top Z)(1-\text{expit}(\eta^\top Z)) Z Z^\top]$ is non-singular and $s_i(\eta) = Z_i (X_{iS} - \text{expit}(\eta^\top Z_i) )$ denotes the score of the $i$th patient in the trial.
This gives
\begin{align*}
    \hat p_{OM} - p = \frac{1}{n^*} \sum_{i=1}^{n^*} \text{expit}(\hat \eta^\top Z_i^*) - \text{expit}(\eta^\top Z_i^*)  +  \frac{1}{n^*} \sum_{i=1}^{n^*} \text{expit}(\eta^\top Z_i^*) - p.
\end{align*}
Now, a first-order Taylor expansion around $\eta$ yields for the first sum
\begin{align*}
    &\frac{1}{n^*} \sum_{i=1}^{n^*} \text{expit}(\hat \eta^\top Z_i^*) - \text{expit}(\eta^\top Z_i^*) \\ &= \Big(\frac{1}{n^*} \sum_{i=1}^{n^*} \big(\text{expit}(\eta^\top Z_i^*)(1-\text{expit}(\eta^\top Z_i^*) \big) Z_i^* \Big)^\top (\hat \eta - \eta) + o_{\mathbb P}(n^{-1/2}) \\
    &= \Gamma^\top (\hat \eta - \eta) + o_{\mathbb P}(n^{-1/2}),
\end{align*}
where the last equality follows from the law of large numbers and 
\begin{align*}
 \Gamma := \mathbb E_Z\big(\text{expit}(\eta^\top Z)(1-\text{expit}(\eta^\top Z))Z \big).
\end{align*}
By applying (OM3) we therefore obtain the representation
\begin{align}
\label{l13}
    \sqrt{n}(\hat p_{OM} - p) = \frac{1}{\sqrt{n}} \sum_{i=1}^n \Gamma^\top \phi_i^{\eta} + \sqrt{\frac{n}{n^*}} \frac{1}{\sqrt{n^*}} \sum_{i=1}^{n^*} \text{expit}(\eta^\top Z_i^*) - p + o_{\mathbb P}(1).
\end{align}
From the proof of asymptotic normality for $\hat \beta$ in \citeSM{andersen1982} it follows that 
\begin{align}
\label{l14}
    \sqrt{n}(\hat \beta_{TS}-\beta_{TS}) &= \frac{1}{\sqrt{n}} \sum_{i=1}^n  e_3^\top \mathcal{I}_0(\beta)^{-1} \delta_i\bigl(\tilde X_i-\bar X(\beta,T_{\mathrm{obs},i})\bigr) + o_{\mathbb P}(1).
\end{align}

Combining this representation with the  expansion \eqref{l13}  gives
\begin{align}
\label{l15}
    \sqrt{n} \Big( \begin{pmatrix}
        \hat \beta_{TS} \\
        \hat p_{OM}
    \end{pmatrix} 
    -
    \begin{pmatrix}
        \beta_{TS} \\
        p
    \end{pmatrix} \Big)
    =  &\frac{1}{\sqrt{n}} \sum_{i=1}^n  \psi_i^{(1)}  + \sqrt{\frac{n}{n^*}} \frac{1}{\sqrt{n^*}} \sum_{i=1}^{n^*}  \psi_i^{(2)} + o_{\mathbb P}(1),
\end{align}

where 

\begin{align*}
\psi_i^{(1)} = 
   \begin{pmatrix}
         e_3^\top \mathcal{I}_0(\beta)^{-1} \delta_i\bigl(\tilde X_i-\bar X( \beta,T_{\mathrm{obs},i})\bigr) \\
        \Gamma^\top \phi_i^{\eta}
    \end{pmatrix}  ~,~~
   \psi_i^{(2)}= \begin{pmatrix}
         0 \\
        \text{expit}(\eta^\top Z_i^*) - p
    \end{pmatrix}.
\end{align*}
Following the arguments in \citeSM{andersen1982} we obtain for the first sum in \eqref{l15}
\begin{align*}
    \frac{1}{\sqrt{n}} \sum_{i=1}^n \psi_i^{(1)} \tod \mathcal{N}_2(0, V_1), \quad \text{as} \ \ n \to \infty
\end{align*}
and, since the sequence $(\psi_i^{(2)})_{i \in \mathbb N}$ is iid, mean-zero with finite second moments, applying the central limit theorem together with (OM2) to the second sum in \eqref{l15} yields
\begin{align*}
    \sqrt{\frac{n}{n^*}} \frac{1}{\sqrt{n^*}} \sum_{i=1}^{n^*}  \psi_i^{(2)} \tod \mathcal{N}_2(0, \kappa V_2), \quad \text{as} \ \ n, n^* \to \infty 
\end{align*}
where the asymptotic covariance matrices are given by $V_j = \mathbb E(\psi_i^{(j)} \psi_i^{(j)\top})$, $j = 1,2$. By assumption (OM1) the two sums are independent, which gives 
\begin{align*}
    \sqrt{n} \Big( \begin{pmatrix}
        \hat \beta_{TS} \\
        \hat p_{OM}
    \end{pmatrix} 
    -
    \begin{pmatrix}
        \beta_{TS} \\
        p
    \end{pmatrix} \Big) \tod \mathcal{N}_2(0, V), \quad \text{as} \ \ n, n^* \to \infty,
\end{align*}
where $V = V_1 + \kappa V_2$. The derivation of the estimator $\hat V$ in \eqref{vestim} is now straightforward.
\end{proof}

\bibliographystyleSM{chicago}
\bibliographySM{references_appendix}

\end{document}